\documentclass[letterpaper,twocolumn,10pt]{article}
\usepackage{usenix-2020-09}
\usepackage{amsfonts}
\usepackage{amsthm}
\usepackage{mathtools}
\usepackage{multirow}
\usepackage{xcolor}
\usepackage[frozencache,cachedir=.]{minted}
\usepackage{subfigure}
\usepackage{xspace}
\usepackage{mathtools}  
\usepackage{enumitem}
\usepackage[group-separator={,},group-minimum-digits={3}]{siunitx}
\usepackage{dblfloatfix}

\newcommand{\Mod}[1]{\ (\mathrm{mod}\ #1)}

\DeclarePairedDelimiter\floor{\lfloor}{\rfloor}

\newtheorem{theorem}{Theorem}[section]

\newtheorem{lemma}[theorem]{Lemma}

\newcommand{\defl}{%
  \ifmmode%
    \Lambda%
  \else%
    $\Lambda$\xspace%
  \fi%
}

\newcommand{\defb}{%
  \ifmmode%
    \Psi%
  \else%
    $\Psi$\xspace%
  \fi%
}

\newcommand{\defc}{%
  \ifmmode%
    \Xi%
  \else%
    $\Xi$\xspace%
  \fi%
}

\begin{document}
\date{}

\title{Swing: Short-cutting Rings for Higher Bandwidth Allreduce}

\author{
{\rm Daniele De Sensi}\\
Sapienza University of Rome
\and
{\rm Tommaso Bonato}\\
ETH Zurich
\and
{\rm David Saam}\\
RWTH Aachen University
\and
{\rm Torsten Hoefler}\\
ETH Zurich
} %

\maketitle

\begin{abstract}
The allreduce collective operation accounts for a significant fraction of the runtime 
of workloads running on distributed systems. One factor determining
its performance is the number of hops between communicating nodes, especially on
networks like torus, where a higher number of hops implies multiple messages
being forwarded on the same link, thus reducing the allreduce bandwidth.
Torus networks are widely used on systems optimized for machine learning 
workloads (e.g., Google TPUs and Amazon Trainium devices), as well as on
some of the Top500 supercomputers.
To improve allreduce performance on torus networks we introduce \textit{Swing}, a new 
algorithm that reduces the number of hops between 
communicating nodes by \textit{swinging} between torus directions. 
Our analysis and experimental evaluation show that Swing 
outperforms by up to $3x$ existing allreduce algorithms for vectors ranging
from 32B to 128MiB, on different types of torus and torus-like 
topologies, regardless of their shape and size.
\end{abstract}

\section{Introduction and Motivation}\label{sec:intro}
Allreduce is a collective operation used to aggregate vectors among a set of nodes and to distribute the aggregated result back to them. Allreduce is widely used to perform distributed gradient aggregation when training deep learning models~\cite{bennun2018demystifying}. Studies have shown that it can account for up to 40\% of the total training time~\cite{285119,9138924,265065} and between 19\% and 30\% of the total core hours in MPI jobs running on production supercomputers~\cite{8665758}.

In this work, we design a new allreduce algorithm optimized for torus-like networks. Torus networks are widely used, both on systems optimized for running machine learning (ML) workloads and on some of the top supercomputers~\cite{1592896,bwssss} (e.g., Fugaku uses a 6D torus~\cite{8514929}). Although torus networks are characterized by a lower bisection and global bandwidth compared to other topologies (e.g., Clos), their simplicity and lower cost allow running some workloads such as ML training in a more cost-effective way, since their communication is often arranged as a 3D logical torus~\cite{10.5555/3571885.3571899}.

Seen from a different perspective, torus networks trade off a lower cost for a lower bisection bandwidth, which, however, is enough to train most ML models efficiently~\cite{10.5555/3571885.3571899}. This is the reason why many systems optimized for ML training rely on torus-like topologies. These include, for example, Google TPUs~\cite{jouppi2023tpu} (a TPU v5p pod connects $\sim\num{9000}$ chips on a 3D torus~\cite{tpuv5}), Amazon Trainium devices~\cite{awstranium} (16 chip on a 2D torus), Graphcore IPU-POD~\cite{graphcore} (connecting 64 chips on a 2D torus), Enflame~\cite{enflame} (2D torus).

Researchers proposed several allreduce algorithms~\cite{10.1177/1094342005051521,10.1145/1088149.1088183,10.1145/1810085.1810093}, and the most performing one depends on a combination of vector size, number of nodes, and physical topology~\cite{10.1007/978-3-540-39924-7_38,hoefler-moor-collectives}. Those algorithms perform a predefined number of steps and, at each step, each node
sends and receives data to and from some predetermined nodes. Different trade-offs exist between the number of 
steps to perform (more critical for allreduce on small
vectors) and the total number of bytes it transmits
(more relevant for larger allreduce). However, a third factor that must 
be considered when designing a new collective algorithm is the number of hops 
between communicating nodes~\cite{10.1177/1094342005051521,10.1016/j.jpdc.2008.09.002,10.1145/2686882,10.1007/978-3-540-39924-7_38}. This is particularly relevant on networks that do not provide full bisection bandwidth such as torus, since the higher the number of hops, the higher the number of flows sharing the same links. 

\begin{figure}[h]
    \centering
    \includegraphics[width=\linewidth]{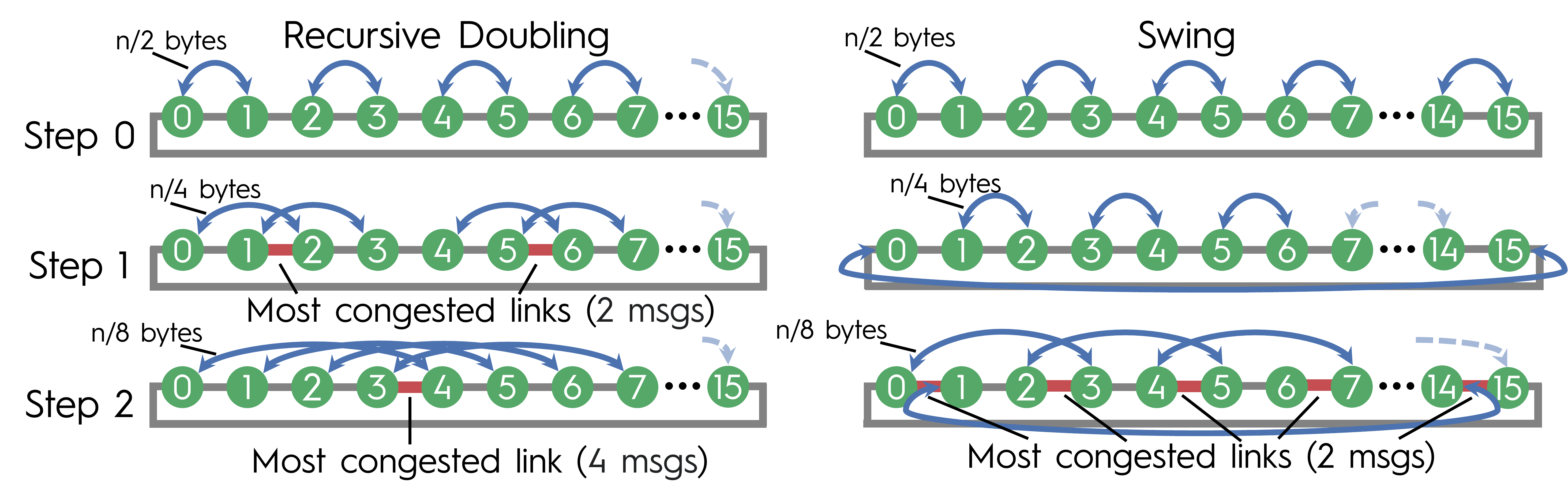}
    \caption{First three steps of the recursive doubling and Swing allreduce algorithms on a 1D torus with $16$ nodes.}
    \label{fig:posterchild_tiny}
\end{figure}

We show the importance of the number of hops in the allreduce through an example.
In Fig.~\ref{fig:posterchild_tiny}, we show a $16$ node 1D torus (we only show a subset of the nodes since the communications are symmetric).
We assume minimal (i.e., shortest path) routing and we show the communications performed by the
bandwidth-optimal recursive doubling algorithm~\cite{10.1177/1094342005051521} 
(also known as \textit{Rabenseifner} algorithm~\cite{10.1007/978-3-540-24685-5_1}, which we describe more in detail in Sec.~\ref{sec:background:torus:recdoubb}),
and by the \textit{Swing} algorithm (that we propose in this work).

Both algorithms perform the same number of steps (we show only the first three
for simplicity). We denote with $n$ the number of bytes of the allreduce vector. In the first step, in both algorithms, each
node $r$ sends $n/2$ bytes to node $q = r \text{ XOR } 1$ (and receives
$n/2$ bytes from it). 
In the second step, however, in the recursive doubling, each node 
$r$ sends $n/4$ bytes to node $q = r \text { XOR } 2$ (two hops distant),
whereas in the Swing algorithm, each node still sends $n/4$ bytes of data,
but with the other neighbor (one hop distant). 

Although both algorithms transmit the same number of bytes,
two different messages cross the same link in the recursive doubling.
For example, two messages cross the link between nodes $1$ and $2$
and that between nodes $5$ and $6$. 
As a consequence, in the worst case all nodes \textit{{transmit data at 
most at half the bandwidth}} of the link between $1$ and $2$, thus slowing
down the entire allreduce operation.
Instead, in the Swing algorithm, each node can still 
transmit at full bandwidth because, in this example, in the second step
each link is crossed at most by one message per direction. Something similar also happens in the third step. 
Indeed, when using Swing at most two messages cross each link instead of the four messages crossing
the link between nodes $3$ and $4$ in recursive doubling.

It is thus clear how even if two different algorithms transmit the same 
number of bytes and perform the same number of steps, they might have different performance, 
depending 
on the network characteristics and the distance between 
communicating nodes. In this example, we have shown an extreme 
case using a 1D torus. However, similar effects can happen
\textit{{on any topology that does not provide full bisection bandwidth}}. 

Although some algorithms (i.e., ring~\cite{10.1016/j.jpdc.2008.09.002,10.1177/1094342005051521} and bucket~\cite{BARNETT1995191,10.1145/2686882, 10.1145/1810085.1810093}) avoid this problem
by having each node communicate with its neighbors only, they perform 
more steps (linear in the number of nodes) and are thus not well-suited for small- and medium-sized vectors. Those are the sizes that, however, are practically used in most machine learning~\cite{10.14778/3415478.3415530} and HPC~\cite{8665758} workloads. Indeed, larger allreduce are split into smaller ones to overlap better computation and communication, especially more when using 3D parallelism in machine learning training~\cite{bennun2018demystifying}.

This work makes the following contributions:
\begin{itemize}[leftmargin=*]
    \item We design a new allreduce algorithm called \textit{Swing}, which performs a logarithmic number of steps and transmits the minimal number of bytes while reducing the distance between communicating nodes compared to other known algorithms designed for small- and medium-sized vectors (Sec.~\ref{sec:swing} and Sec.~\ref{sec:swing:multid}).
    \item We evaluate Swing on different torus and torus-like topologies (e.g., HammingMesh~\cite{10.5555/3571885.3571899} and HyperX~\cite{10.1145/1654059.1654101,10.1145/3295500.3356140}), by comparing it with the best state-of-the-art algorithms (Sec.~\ref{sec:evaluation}). Our evaluation shows that \textit{{Swing outperforms the other existing algorithms for allreduce on vectors ranging from 32B to 128MiB on different torus-like topologies, and regardless of their shape and size}}. We show that Swing outperforms the best-known algorithm up to 2.2x on square torus with \num{4096} to \num{16384} nodes and up to 3x on rectangular tori and HyperX with \num{4096} nodes.
\end{itemize}

\section{Background}\label{sec:background}
\subsection{Targeted Collectives}
We briefly introduce the \textit{reduce-scatter} and \textit{allgather} 
collectives since for medium and large vectors, Swing allreduce algorithm 
executes a reduce-scatter followed by an allgather (similarly to the 
Rabenseifner algorithm~\cite{10.1007/978-3-540-24685-5_1}).
In the reduce-scatter, the compute nodes
reduce vectors (one per node) using a reduction operation (e.g., addition) and shard the resulting vector across all the nodes. 
In the allgather, each node provides a vector that they concatenate and distribute to all the nodes.

If the vector contains a number of elements larger or equal than $p$, 
the nodes can run the allreduce as a reduce-scatter followed by an 
allgather~\cite{10.1177/1094342005051521}. I.e.,  they aggregate
vectors coming from all the nodes and distribute back the resulting
vector. Although for space reasons we mainly target the allreduce, 
Swing can also be used for performing reduce-scatter and
allgather collectives, as well as any other collective operation where recursive doubling or binomial tree can be used  (e.g., broadcast and reduce)~\cite{10.1007/978-3-540-39924-7_38}.

\subsection{Notation and Model}
We consider $D$-dimensional tori of size $\{d_0, d_1, \ldots, d_{D-1}\}$, and we 
denote with $p$ the number of nodes in the torus, i.e., $p = d_0 \cdot d_1 \cdot \ldots \cdot d_{D-1}$. 
We assume the collectives run on $p$ nodes, that ranks are mapped to nodes linearly and, 
without loss of generality, to have one process (or \textit{rank}) per node.
We assume that each node in the network has $2\cdot D$ ports, that each link is bidirectional, and
that each node can send  $2\cdot D$ messages and receive $2\cdot D$ messages concurrently (one send 
and receive per port). We also assume the network forward packets using minimal adaptive routing, and that does not have any hardware support to accelerate collective 
operations (e.g., in-network aggregation~\cite{10.1145/3458817.3476178,Graham2020ScalableHA,10.1145/1088149.1088183}).

To support the description of the different algorithms, we model their performance 
with the commonly used latency and bandwidth model~\cite{10.1145/173284.155333,ALEXANDROV199771,1420226}.
We model the communication time $T(n)$ to send $n$ bytes in a point-to-point 
communication as $T(n) = \alpha + n\beta$, where $\alpha$ represents the 
\textit{latency} (i.e., the time for the first byte to reach the destination), 
and $\beta$ the time to transmit a single byte (it can be seen as the inverse 
of the \textit{bandwidth}). When modeling collective operations involving data reduction, researchers also consider
an additional $\gamma$ term to model the aggregation cost. 
To avoid burdening the notation, we do not 
model this term explicitly since Swing is no worst than the other 
algorithms in that regard and most implementations overlap the aggregation with the communication~\cite{1243181,ompiallreduce}. 

\begin{table}
\footnotesize
\centering
\begin{tabular}{c|c}
 \textbf{\textsc{Name} }& \textbf{\textsc{Description}} \\  \hline
 $D$  & Number of torus dimensions \\ \hline
 $d_0, \ldots, d_{D-1}$ & Size of each dimension \\ \hline
 $n$ & Size of the vector to reduce \\ \hline
 $p$ & Number of nodes in the network \\ \hline
 \defl & Latency deficiency \\ \hline
 \defb & Bandwidth deficiency \\ \hline
 \defc & Congestion deficiency \\ \hline
 $\delta(s)$ & Number of hops between $2$ communicating nodes at step $s$ \\ \hline
 $\rho(s)$ & $\sum_{i=0}^{s} -2^i$ \\ \hline
 $\pi(r, s)$ & The node with which node $r$ communicates at step $s$ \\ 
\end{tabular}
\caption{Variables and functions used in our modeling.}
\label{tab:varfun}
\end{table}

Collective operations involve multiple communication steps. Previous
works proved that
the allreduce requires at least $\log_2{p}$ steps and the transmission
of at least $2\frac{p-1}{p}n \approx 2n$ bytes of data~\cite{BARNETT1995191}. 
Hence, the optimal allreduce time 
can be modeled as $T(n) = \alpha\log_2{p} + \beta2n$. However,
because each node has $2D$ ports, bandwidth-optimal algorithms 
distribute the data equally across all the ports~\cite{10.5555/3571885.3571899,BARNETT1995191,10.1145/2686882, 10.1145/1810085.1810093}, and we can
model the allreduce 
time as $T(n) = \alpha \log_2{p} + \beta\frac{n}{D}$
In practice, however, algorithms have some \textit{deficiency} compared to optimal, 
either in the latency or the bandwidth terms (or in both). 

We consider three different type of deficiencies: i) \textit{latency deficiency} (\defl) i.e., 
how much the latency is higher than the optimal; ii) \textit{algorithmic bandwidth 
deficiency} (\defb), i.e., how 
many more bytes does the algorithm transmit; iii) \textit{congestion bandwidth deficiency} (\defc) i.e., 
what is the bandwidth slowdown caused by multiple messages of the same collective being forwarded 
on the same link (as discussed in Sec.~\ref{sec:intro}). We can see deficiencies as multiplicative
factors that denote how much an algorithm is distant from the optimal (e.g., a latency deficiency
of one means that the algorithm is latency optimal), and we can thus model the allreduce time as:
\begin{equation}
T(n) = \log_2{p} \cdot \alpha \cdot \defl + \frac{n}{D}\beta \cdot \defb \cdot \defc
\end{equation}

For brevity, we refer to $\defb$ as \textit{bandwidth deficiency} and to $\defc$ as \textit{congestion deficiency}.
While \defl and \defb only depend on the algorithm, the congestion deficiency
\defc depends on the network. 
To simplify the notation and the discussion, in the following, we are only going to model
the three deficiencies $\defl$, $\defb$, and $\defc$.
We summarize the variables we use in our modeling in Table~\ref{tab:varfun}.

\subsection{State-of-the Art Algorithms}\label{sec:multidbackground}
In the following, we review and model the main allreduce algorithms
for multidimensional torus, and we summarize in Table~\ref{tab:deficiency:torus} 
their deficiencies, as well as those of the Swing algorithm.
\subsubsection{Hamiltonian Rings}\label{sec:background:ring}
Ring allreduce algorithm~\cite{10.1016/j.jpdc.2008.09.002,10.1177/1094342005051521} runs a reduce-scatter followed by an allgather. Each node splits its data into $p$ equally sized blocks. For both the reduce-scatter and the allgather, the algorithm performs $p-1$ steps. Nodes are arranged in a ring, and at each step, each node sends a block to its right neighbor and receives a block from its left neighbor. Because the algorithm performs $2(p-1) \approx 2p$ the latency deficiency is $\defl = \frac{2p}{\log_2{p}}$.

The algorithm sends $\approx 2n$ bytes ($n$ in the reduce-scatter and $n$ in the allgather). On multiport networks (assuming $2D$ ports), it splits the data into $2D$ parts (of $n/2D$ bytes each) and runs one ring algorithm on each part (each sending and receiving to and from a different port). Since it sends a minimal number of bytes and uses all the ports, the algorithm has no bandwidth deficiency ($\defb = 1$). 

Moreover, because each node only communicates with neighbors on a 1D torus, the algorithm does not have congestion deficiency since each link is used by at most one communication in each direction. In the version optimized for the 2D torus, the four rings that run in parallel are mapped into two edge-disjoint Hamiltonian cycles (one per direction)~\cite{10.5555/3571885.3571899} so that each link is still used by at most one communication per direction (thus $\defc = 1$). To our knowledge, this algorithm does not work for $D > 2$. Moreover, the algorithm can build the two edge-disjoint Hamiltonian cycles on an $r \times c$ 2D torus only if $r = c \cdot k, \: \: k \geq 1$ and the greatest common divisor between $r$ and $c-1$ is $1$, which limits the applicability of the algorithm.

\subsubsection{Latency-Optimal Recursive Doubling}\label{sec:background:torus:recdoubl}
The latency-optimal recursive doubling algorithm~\cite{10.1177/1094342005051521} executes $\log_2{p}$ steps (thus it has a latency deficiency $\defl=1$). At each step $s$ (we denote steps starting from $0$), each node $r$ sends its vector to node $q = r \text{ XOR } 2^s$ (assuming $p$ is a power of $2$) and receives $q$'s vector, which aggregates with its own before moving to the next step. When running on a torus, if the size of each dimension is a power of two, it can keep a shorter distance between communicating nodes by communicating in a different dimension at each step, as shown in Fig.~\ref{img:multid_rd}. 

\begin{figure}[h]
    \centering
    \includegraphics[width=\linewidth]{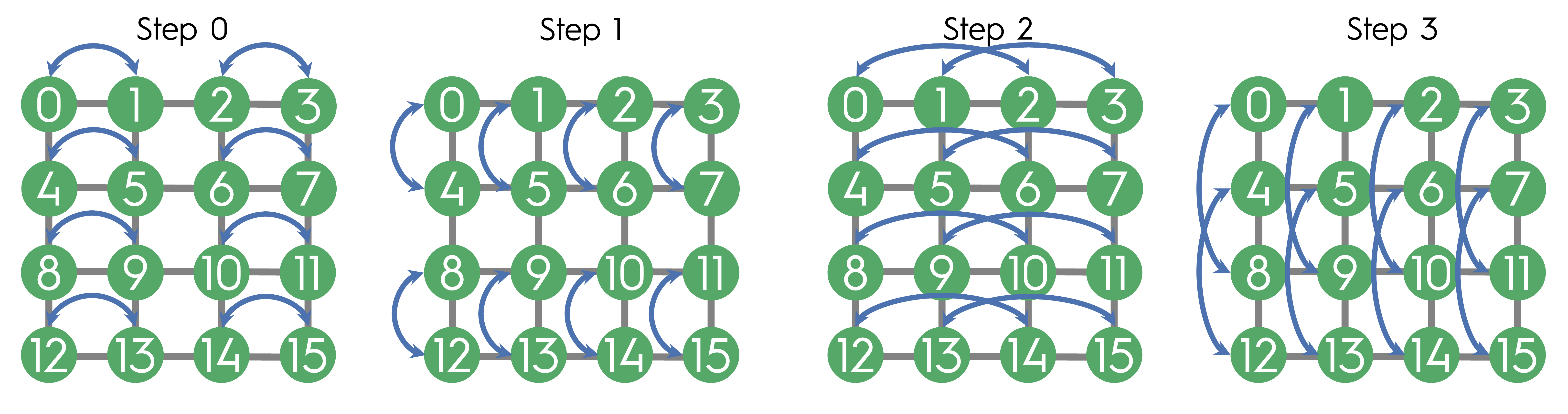}
    \caption{Example of recursive doubling on a $4x4$ torus. Wrap-around links are not shown.}
    \label{img:multid_rd}
\end{figure}

Each node transmits $n \log_2{p}$ bytes of data. To our knowledge, no multiport versions of this algorithm exist, and we model its bandwidth deficiency as $\defb = D \log_2{p}$. Each node communicates with $D$ nodes at distance $2^i$ (one per dimension), with $0 \leq i < \frac{\log_2{p}-1}{D}$. For this algorithm, the number of messages forwarded on the most congested link is equal to the distance between communicating nodes\footnote{The only exception to this is the last step in each dimension because each node
can reach its peer with two different minimal paths. For example, in Fig.~\ref{img:multid_rd}, in step $2$ node $0$ can send half of the
packets directed to $2$ to the right and half to the left on the wrap-around
link (not shown). However, this is negligible for large enough networks.}.
We can thus estimate the congestion deficiency as the sum of the distance between communicating nodes over all the steps, i.e., $\defc = D\sum_{i=0}^{\frac{\log_2{p}-1}{D}} 2^i \leq 2D\sqrt[D]{p}$.

\paragraph{Non-power-of-two}
If the size of a dimension is not a power of two, some extra steps are needed. One possible solution consists in reducing the number of nodes to the largest power-of-two $p' < p$~\cite{doi:10.1142/S0129626493000046,doi:10.1142/S012962649300037X}. Before starting the allreduce, each node in the range $(p', p-1)$ sends its data to a node in the range $(0, p'-1)$. Then, the first $p'$ nodes run the allreduce and, when completed, send the reduced data to the nodes in the range $(p', p-1)$. This algorithm increases all the deficiencies compared to the case where $p$ is a power of two, but we do not model this explicitly for brevity.

\begin{table}
\scriptsize
\centering
\begin{tabular}{p{1.65cm}|p{1.5cm}|p{1.7cm}|p{.25cm} p{.25cm} p{.25cm}}

 \multirow{2}{*}{\textbf{\textsc{Algorithm}}} & \multirow{2}{*}{\textbf{\textsc{Lat. Def. (\defl)}}}        & \multirow{2}{*}{\textbf{\textsc{Band. Def. (\defb)}}} & \multicolumn{3}{c}{\textbf{\textsc{Cong. Def. (\defc)}}} \\ %
 & & & $D{=}2$ & $D{=}3$ & $D{=}4$ \\ \hline 
 \textsc{Ring}      & $\frac{2p}{\log_2{p}}$            & $1$ & \multicolumn{3}{c}{$1$} \\  
 
 \textsc{Rec.Doub. (L)}      & $1$    & $ D \log_2{p}$ & \multicolumn{3}{c}{$2D\sqrt[D]{p}$} \\  
 \textsc{Rec.Doub. (B)} & $2$    & $2D$ &  \multicolumn{3}{c}{$\frac{2^D - 1}{2^D - 2}$} \\  
 \textsc{Bucket}    & $\frac{2D\sqrt[D]{p}}{\log_2{p}}$ & $1$ & \multicolumn{3}{c}{$1$}\\  \hline  
 \textsc{Swing (L)}     & $1$    & $D\log_2{p}$ &  \multicolumn{3}{c}{$\frac{4}{3}D\sqrt[D]{p}$}\\
 \textsc{Swing (B)}     & $2$    & $1$ & $1.19$ & $1.03$ & $1.008$\\
\end{tabular}
\caption{Algorithms deficiencies on $D$-dimensional torus with $D \geq 2$. (L) and (B) stands for latency-optimal and bandwidth-optimal (\textit{-optimized} for recursive doubling) respectively.}
\label{tab:deficiency:torus}
\end{table}

\subsubsection{Bandwidth-Optimal Recursive Doubling}\label{sec:background:torus:recdoubb}
The classic bandwidth-optimal recursive doubling algorithm (also known as Rabenseifner algorithm~\cite{10.1177/1094342005051521}) performs an allreduce as a reduce-scatter followed by an allgather. The reduce-scatter and the allgather use recursive doubling (each performing $\log_2{p}$ steps). Differently from the latency-optimal algorithm, each node divides the data into $p$ blocks $\{b_0, \ldots, b_{p - 1}\}$, each of size $\frac{n}{p}$. 

At each step, the reduce-scatter halves the size of the transmitted data and doubles the distance between communicating nodes. Thus, each node transmits $n$ bytes of data in the reduce-scatter. The allgather works similarly but reverses the communication pattern, doubling the size of the transmitted data at each step and halving the distance between communicating nodes. The allreduce executes $2\log_2{p}$ steps ($\defl = 2$), and transmits $2n$ bytes of data. 

Sack et al. optimized the algorithm for torus networks~\cite{10.1145/2686882}, similarly to what we described for the latency-optimal algorithm, to reduce its congestion deficiency to $\defc = \frac{2^D - 1}{2^D - 2}$ for $D > 1$~\cite{10.1145/2686882}. However, to our knowledge, no multiport versions of this algorithm exist, and its bandwidth deficiency is $\defb = 2D$. Hence, for torus networks we consider this algorithm as \textit{bandwidth-optimized} rather than \textit{bandwidth-optimal}. 

\paragraph{Non-power-of-$2$}
If the number of nodes $p$ is not a power of $2$, the algorithm performs some extra steps~\cite{10.1007/978-3-540-30218-6_13, 10.1177/1094342005051521} that increase the latency and reduce the bandwidth (because transmit extra data). Allreduce implementations can use either similar techniques like those described for the latency-optimal recursive doubling or more sophisticated ones like the \textit{3-2 elimination} technique~\cite{10.1007/978-3-540-30218-6_13} (which increases the bandwidth deficiency to $3/2$). However, we found no reference of adaptations to torus networks.
\subsubsection{Bucket Algorithm}\label{sec:background:bucket}
To simplify the exposition, we first describe the algorithm for a $a \times a$ 2D torus (with $a \cdot a = p$)~\cite{BARNETT1995191}. Each node runs a ring reduce-scatter with the other $a-1$ nodes on the same row. This requires $a-1$ steps, and each node transfers $n\frac{a-1}{a}$ bytes. Then a ring reduce-scatter with the other $a-1$ nodes on the same column, but only on the data already reduced at the previous step (of size $\frac{n}{a}$). Then, each node runs an allgather with all the nodes on the same column and then with all those on the same row.

On a $D-$dimensional torus, the algorithm performs $D$ reduce-scatter followed by $D$ allgather (each on $\sqrt[D]{p}$ nodes on a square torus). Because it runs $2D\sqrt[D]{p}$ steps, the latency deficiency is $\defl=\frac{2D\sqrt[D]{p}}{\log_2{p}}$. To use all the $2\cdot D$ ports, the algorithm splits the data into $2\cdot D$ parts and concurrently runs $2 \cdot D$ bucket algorithms (one for each part)~\cite{10.1145/2686882, 10.1145/1810085.1810093}. Since the algorithm sends the minimal number of bytes and uses all the ports evenly, the bandwidth deficiency is $\defb = 1$. Each of the $2\cdot D$ bucket algorithms starts from a different dimension and direction so that, at each step, each link is used by at most one ring per direction (i.e., the congestion deficiency is $\defc=1$). 

\subsubsection{Other Approaches} 
\paragraph{Topology-Specific Algorithms} Researchers proposed several allreduce algorithms~\cite{10.1007/978-3-540-39924-7_38,BARNETT1995191,kolmakov2020generalization,RUEFENACHT201724,JOCKSCH2021102812,10.1007/978-3-540-24685-5_1}, some of which optimized for specific topologies~\cite{10.1145/1088149.1088183,9644896,10.1145/3524059.3532380}. In this work, we focus on those explicitly designed for torus networks, since they are characterized by a lower congestion deficiency. 

\paragraph{Automatic Generation of Collective Algorithms} Some approaches use linear programming formulations for finding the best collective algorithm given a network specification and the size of the collective~\cite{shah2022taccl, 10.1145/3437801.3441620}. However, this requires solving an NP-hard problem that grows exponentially with scale. Finding a solution for 128 nodes requires up to 11 hours~\cite{shah2022taccl}, and a new solution might need to be found when changing the number of nodes or the size of the collective. This makes generating collective algorithms for large systems like the \num{9000} nodes Google's TPU v5p pod~\cite{tpuv5} challenging or even impossible. On the contrary, Swing can seamlessly run on any number of nodes. Moreover, unlike Swing, some of these solutions do not explicitly model the congestion deficiency.

\paragraph{Topology Reconfiguration and In-Network Compute} Other solutions improve allreduce performance by re-configuring the network topology according to the specific traffic pattern~\cite{285119}. Swing is orthogonal to these approaches and, by reducing the network congestion, can make the expensive network re-configurations less frequent. Last, some algorithms exploit in-network compute capabilities of programmable switches~\cite{SwitchML,265053,10.1145/3452296.3472904,DESENSI202470} to aggregate data directly in the network, reducing network traffic and improving performance. However, unlike Swing, these solutions require specific switches to be deployed in the network, whereas Swing can seamlessly run on any network.

\section{Swing Design}\label{sec:swing}
By analyzing the algorithms we described, we observe different tradeoffs. The latency-optimal recursive doubling has the lowest latency deficiency. It is thus more suited to small allreduce, where the number of steps executed by the algorithm, rather than the total number of transmitted bytes, dominates the runtime. On the other hand, ring and bucket algorithms are characterized by the lowest bandwidth and congestion deficiency, and we expect them to perform better on large allreduce. The bandwidth-optimized recursive doubling lies somewhere in between since it has a higher bandwidth and congestion deficiency but a lower latency deficiency and would perform better for medium-sized vectors.

With the Swing algorithm, we aim at designing an algorithm with a congestion deficiency \defc lower than the bandwidth-optimized recursive doubling algorithm by reducing the distance between communicating nodes. We also aim to reduce bandwidth deficiency \defb by using all the $2D$ available ports. To simplify the exposition, we first discuss the design of the bandwidth-optimal Swing algorithm on a 1D torus, assuming the number of nodes $p$ is a power of $2$ (Sec.~\ref{sec:swing:p2}). Then, we extend it to any number of nodes (Sec.~\ref{sec:swing:nonp2}), and describe its design for tors with more than one dimension (Sec.~\ref{sec:swing:multid}).

\subsection{Algorithm Design}\label{sec:swing:p2}
\subsubsection{Bandwidth-optimal Algorithm}
We describe in the following the design of the Swing algorithm, and we formally prove its correctness in Appendix~\ref{sec:proof}.
The bandwidth-optimal Swing algorithm runs a reduce-scatter followed by an allreduce. In the reduce scatter, at step $s$ (starting from $0$), each node $r$ communicates with a node $\pi(r, s)$ such that:
\begin{equation}\label{eq:pi}
\pi(r, s) = 
\begin{cases}
    r + \rho(s) \bmod p, & \text{if } r \text{ is even} \\
    r - \rho(s) \bmod p, & \text{if } r \text{ is odd} \\
\end{cases}
\end{equation}

Where $\rho(s) = \sum_{i=0}^{s} (-2)^i = \frac{1 - (-2)^{s+1}}{3}$. This selection of the communicating peer leads to a communication pattern like the one shown in Figure~\ref{fig:posterchild_tiny} for a $16$ nodes 1D torus. We observe how, at each step, the communicating peer of each node \textit{swings} from left to right and vice versa (hence the algorithm's name). \textit{{Intuitively, unlike recursive doubling, each node reaches distant nodes in fewer hops by short-cutting the ring}}. 

More precisely, at each step, each node communicates with a node at a distance $\delta(s)$, with:
\begin{align}
\begin{split}
\delta(s) &= |\rho(s)| = \Big|\sum_{i=0}^{s} -2^i\Big| = \frac{2^{s+1} - (-1)^{s+1}}{3} \leq \\
&\leq \frac{2^{s+1} + 1}{3} < 2^s + \frac{1}{3}\nonumber
\end{split}
\end{align}
Because $\delta(s)$ is always a natural number, we have $\delta(s) \leq 2^s$ (it is strictly smaller for $s > 1$). Hence, Swing has a lower congestion deficiency than recursive doubling because of the lower distance between communicating nodes (we estimate precisely the congestion deficiency in Sec.~\ref{sec:swing:multid}).

For simplicity, we first describe the reduce-scatter algorithm using only one port and extend it to use all the $2D$ ports. In the reduce-scatter each node splits data into $p$ blocks $\{b_0, b_1, \ldots, b_{p - 1}\}$, each of size $\frac{n}{p}$. Each node $r$ runs $\log_2{p}$ steps, communicating at each step $s$ with the node $\pi(r, s)$ and halving the size of the transmitted data. At the end of 
the reduce-scatter, we want each node $r$ to have the fully aggregated block $b_r$.

To do so, data transmitted from $r$ to $q$ includes the block $b_q$, plus all the blocks that $q$ will transmit to other nodes in the subsequent steps. The allgather works similarly, but each node selects its peer in the reverse order, thus communicating first with the more distant ones. In the first step, each node $r$ sends its block $b_r$, doubling the transmitted data's size at each step (data transmitted from $r$ to $q$ includes all the blocks that $r$ gathered until step $s$).
Because the algorithm performs $2\log_2{p}$ steps, its latency deficiency is $\defl{=}\frac{\log_2{p}}{2}$. Because it transmits the minimal number of bytes and uses all the ports (as we will show in Sec.~\ref{sec:swing:multid}), its bandwidth deficiency is $\defb=1$. We estimate the congestion deficiency in Sec.~\ref{sec:swing:multid} when describing the algorithm for torus with more than one dimension.

We summarize the algorithm in Listing~\ref{lst:1d} for reduce-scatter (the algorithm for allgather is analogous). The function \texttt{get\_rs\_idxs} computes the indexes of the data blocks that a given node \texttt{r} must send at step \texttt{step}, and relies on the function \texttt{pi} we defined in Eq.~\ref{eq:pi}. Then, the \texttt{reduce\_scatter} function executes $\log_2{p}$ steps, and at each step, computes the bitmaps $\texttt{blocks\_s}$ and $\texttt{blocks\_r}$ denoting the blocks of data that must be sent and received. Last, it sends and receives those blocks. 

\begin{listing}[!ht]
\begin{minted}[obeytabs=true,tabsize=2,fontsize=\footnotesize,frame=lines]{python}
def get_rs_idxs(r, step, p, blocks):
  if step >= log2(p): return
  for s in range(step, int(log2(p))):
    peer = pi(r, s, p)
    # Set to 1 the node I directly reach
    blocks[peer] = 1 
    # and those that it will reach
    get_rs_idxs(peer, s+1, p, blocks) 
    
def reduce_scatter(r, p, data):    
  for s in range(0, int(log2(p))):
    blocks_s = blocks_r = [0]*p
    dest = pi(r, s)
    get_rs_idxs(r, s, p, blocks_s)
    get_rs_idxs(peer, s, p, blocks_r)
    # Send blocks where blocks_s[i]=1, 
    # recv blocks where blocks_r[i]=1
    sendrecv(dest, data, blocks_s, block_r)        
\end{minted}
\caption{Swing reduce-scatter pseudocode.}
\label{lst:1d}
\end{listing}

We can transmit non-contiguous data using, for example, MPI \textit{indexed} datatypes. However, because communicating non-contiguous data can introduce overhead~\cite{gropp1999improving,DBLP:conf/pvm/SchneiderKH13}, in the allreduce, we logically remap the blocks (i.e., without any actual memory movement) so that each node sends contiguous data. Indeed, even if the algorithm shuffles the block in the reduce-scatter, they are eventually placed again in the proper order in the buffers by the allgather. Moreover, by sending contiguous data, we also guarantee that the algorithm works with non-commutative reduction operators~\cite{10.1007/978-3-540-30218-6_13,10.1007/978-3-540-39924-7_38}.

\subsubsection{Latency-optimal Algorithm}
The latency-optimal Swing algorithm uses the same communication pattern as the bandwidth-optimal one, but instead of running a reduce-scatter followed by an allgather, at each step each node exchanges its entire vector with that of its peer (similarly to the latency-optimal recursive doubling). The algorithm only requires $\log_2{p}$ steps ($\defl=1$) but transmits $n \cdot \log_2{p}$ bytes ($\defb = D\log_2{p}$ because the algorithm uses all the ports).  We estimate the congestion deficiency in Sec.~\ref{sec:swing:multid} when describing the algorithm for torus with more than one dimension.

\subsection{Non-power-of-two Nodes}\label{sec:swing:nonp2}
When $p$ is even but not a power of $2$, some nodes can receive the same block of data twice (one of which in the last step, see Appendix~\ref{sec:proof:nonpower}). Thus, in that case, it is enough for each node not to send the same data block twice. Because no extra data is sent compared to the power of two cases, deficiencies do not increase. 

If $p$ is odd, we run the algorithm on $p-1$ nodes, while node $p-1$ at each step sends $(p-1)/ 2^s$ of its blocks to the corresponding $(p-1) / 2^s$ nodes (and receiving from those its block). We show this through an example in Fig.~\ref{img:non-pow-2} for a 1D torus with $7$ nodes (we only show the first two steps). The first $6$ nodes run the algorithm for even nodes as usual. At step $0$ the last node sends (and receives) $\frac{n}{7}$ bytes to nodes $0$, $1$, and $2$. At step $1$, the last node sends $\frac{n}{7}$ bytes to node $3$ and $4$, and in the last step $\frac{n}{7}$ bytes to node $5$. This slightly increases the bandwidth deficiency (by a $1/p$ additive factor).

\begin{figure}[h]
    \centering
    \includegraphics[width=0.7\linewidth]{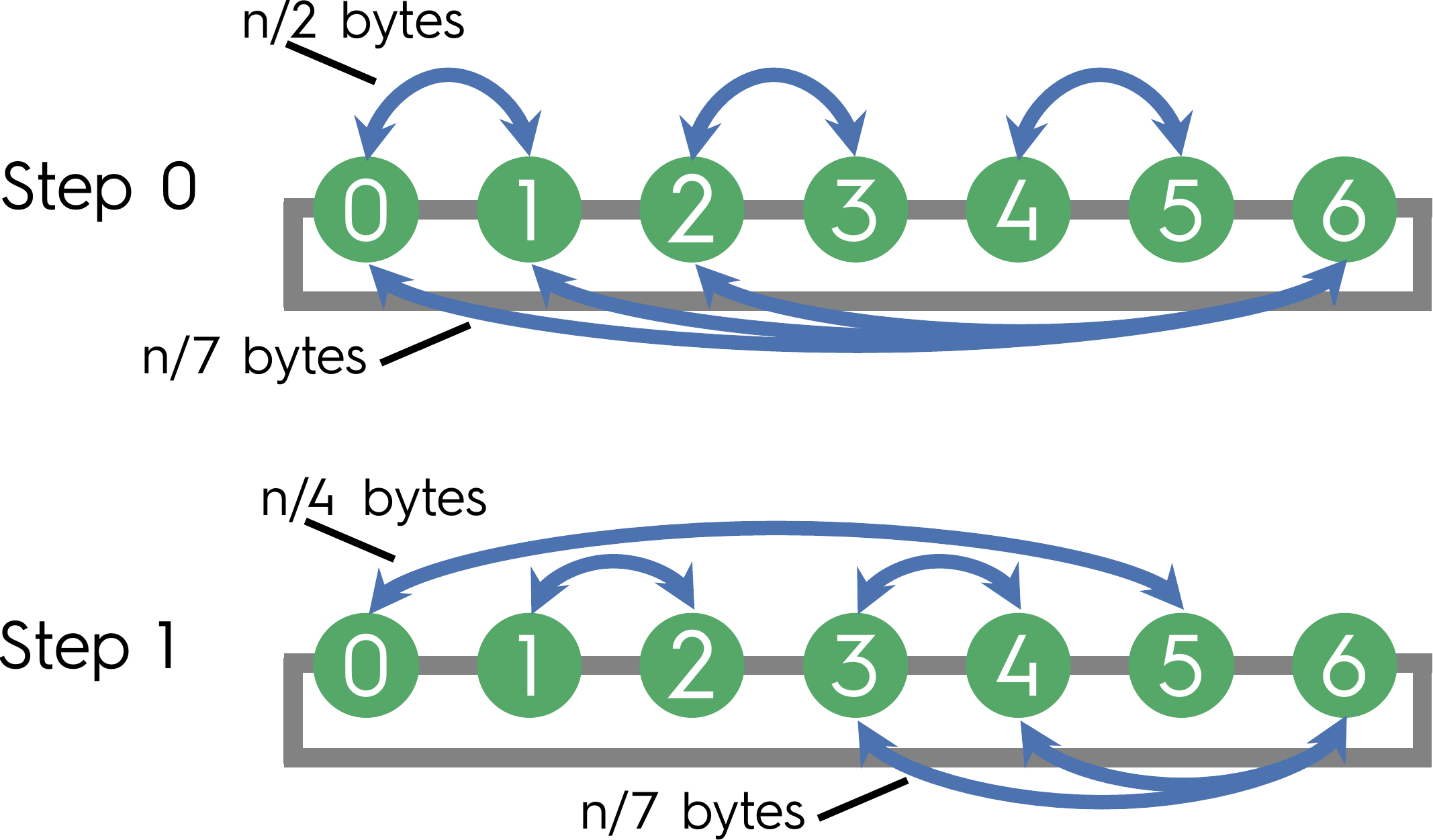}
    \caption{First 2 step of the Swing algorithm on a 1D torus with 7 nodes.}
    \label{img:non-pow-2}
\end{figure}

\section{Design for Multidimensional Tori}\label{sec:swing:multid}
In Sec.~\ref{sec:swing:p2}, we described the design of the Swing algorithm for 1D torus. We now discuss how to extend it to square (Sec.~\ref{sec:swing:multid:square}) and non-square (Sec.~\ref{sec:swing:rectangular}) torus with more than one dimension.

\subsection{Square Tori}\label{sec:swing:multid:square}
Like the recursive doubling algorithms optimized for tori (discussed in Sec.~\ref{sec:background:torus:recdoubl} and Sec.~\ref{sec:background:torus:recdoubb}), in the Swing algorithm (both the latency- and the bandwidth-optimal) each node communicates on one dimension at a time. Formally, at step $s$, each node communicates on the dimension $\omega(s) = s \bmod{D}$. We define with $\sigma(s) = \floor{\frac{s}{D}}$ the step of the algorithm relative to a specific dimension. For example, on a 2D torus, the third step of the algorithm is the second step executed in the first dimension (i.e., because we count steps starting from $0$, $\sigma(2) = \floor{3/2} = 1$).

We denote the coordinates of a node with $(a_0, a_1, \ldots, a_{D-1})$. Then, at step $s$, each node communicates with a node with the same coordinates, except for the coordinate $a_{\omega(s)}$. If $a_{\omega(s)}$ is even, the coordinate $a_{\omega(s)}$ is then replaced with $(a_{\omega(s)} + \delta(\sigma(s))) \bmod{d_{\omega(s)}}$ (if odd, we flip the sign before $\delta(\sigma(s))$).

To use all the $2\cdot D$ ports, Swing splits the data into $2\cdot D$ parts and runs one allreduce on each. To avoid increasing the congestion deficiency, we must guarantee that, at each step, each of these $2 \cdot D$ collectives use different links. Swing runs $D$ of these collectives (that we call \textit{plain} collectives), each starting from a different dimension, using the algorithm described above. Swing runs the remaining $D$ collectives (which we call \textit{mirrored} collectives) with the same approach but starting from the opposite direction than that of the corresponding \textit{plain} collective. By doing so, each of the $2\cdot D$ allreduce uses a different port at each step.

\begin{figure}[h]
    \centering
    \includegraphics[width=\linewidth]{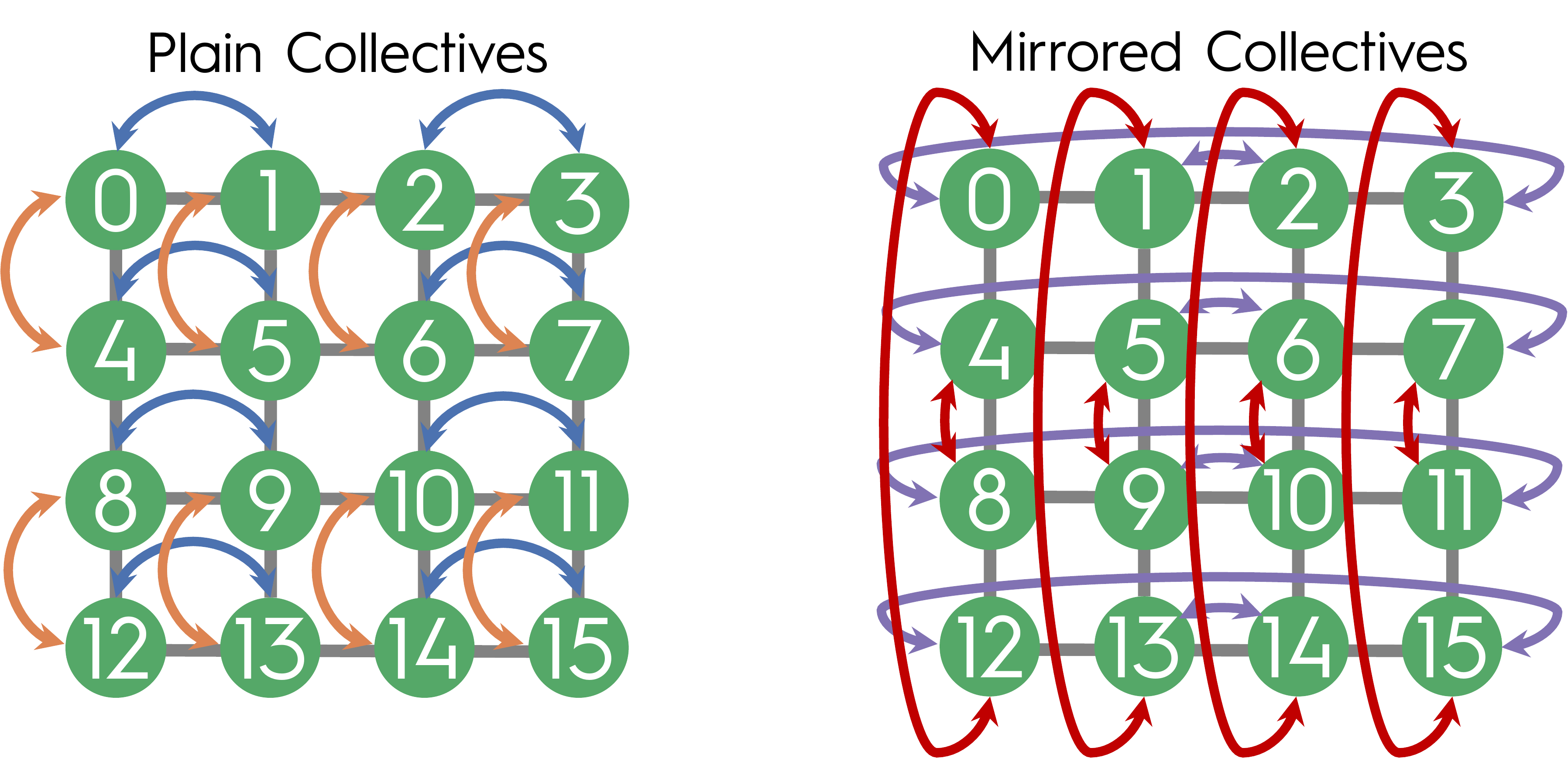}
    \caption{First step of the Swing algorithm on a 4x4 torus.}
    \label{img:multiport}
\end{figure}

We show this through an example in Figure~\ref{img:multiport} on a $4 \times 4$ torus by showing the first step of the Swing algorithm. Node $0$ runs a plain collective on the horizontal dimension exchanging data with node $1$, and one in the vertical dimension with node $4$. The two mirrored allreduce inverts the directions, exchanging data with node $3$ in the horizontal dimension and node $12$ in the vertical one.

Whereas both latency and bandwidth deficiencies are unaffected by the number of dimensions, congestion deficiency decreases with the number of dimensions due to increased bisection bandwidth. Intuitively, the more dimensions the torus has, the more communications the algorithm does with nodes at a closer distance before moving to higher distances. Because Swing halves the data size at each step, when a node needs to communicate with a distant node, data size becomes smaller and is thus less affected by congestion.

We can model the bandwidth term of the Swing allreduce as:
\begin{equation}
\frac{n}{2D}\beta \sum_{s=0}^{\log_2{(p)} - 1} \frac{\delta(\sigma(s))}{2^{s+1}}
\nonumber
\end{equation}

I.e., at each step $s$ Swing halves the size of the data and, because the distance between communicating nodes at step $s$ is $\delta(\sigma(s))$, there is at least one link shared by $\delta(\sigma(s))$, proportionally increasing the time to transmit one byte. 

We can estimate the congestion deficiency by dividing this quantity by $(n\beta)/D$. Instead of deriving a hardly readable closed form, we report in Table~\ref{tab:deficiency:torus} the values for different dimensions and for $p \to \infty$ (since the congestion deficiency increases with $p$). If $D \geq 3$, the Swing algorithm has a congestion deficiency $\defc < 1.003$ (i.e., lower than 3\%). To estimate the latency-optimal version's congestion deficiency, we sum the distances over all the steps, similar to what we did in Sec.~\ref{sec:background:torus:recdoubl} for the latency-optimal recursive doubling algorithm. I.e.:
\begin{equation}
\defc = D\sum_{s=0}^{\frac{\log_2{(p)} - 1}{D}} \delta(s) \leq \frac{4}{3}D\sqrt[D]{p}
\nonumber
\end{equation}

We want to remark that the state-of-the-art latency-optimal and bandwidth-optimized recursive doubling algorithms described in Sec.~\ref{sec:background:torus:recdoubl} and Sec.~\ref{sec:background:torus:recdoubb} only use one port. In principle, we could extend them to use $2D$ ports by using the same approach we used for Swing, running $D$ plain and $D$ mirrored collectives. However, as we show in Sec.~\ref{sec:evaluation:2d}, they will perform strictly worse than Swing. Indeed, while mirroring decreases their bandwidth deficiencies, their congestion deficiencies are still higher than that of Swing due to the higher distance between communicating nodes.

\subsection{Non-Square Tori}\label{sec:swing:rectangular}
If not all the dimensions have equal size, the algorithm completes all the steps in one dimension while there are still steps to execute in other dimensions. If $d_{min}$ is the smallest dimension, for the first $D \cdot \log_2(d_{min})$ the algorithm behaves exactly like in a $d_{min} \times \ldots \times d_{min}$ torus. After that, no data is sent anymore on that dimension, and the algorithm proceeds on the remaining ones. However, from that point on, it does not use all the available ports. 

Indeed, since data is not transmitted anymore into one of the dimensions, the ports on that dimension are not used. If $d_{min}$ is large enough, this has a limited impact because the size of the transmitted data decreases after each step. We show this through an example in Fig.~\ref{img:rectangular} on a $2 \times 4$ torus where to not clutter the figure, we report only the communications performed by node $0$. In step $2$ (the last step), all the $4$ collectives communicate on the horizontal dimension since each node has already reached all the nodes in their column.

\begin{figure}[h]
    \centering
    \includegraphics[width=\linewidth]{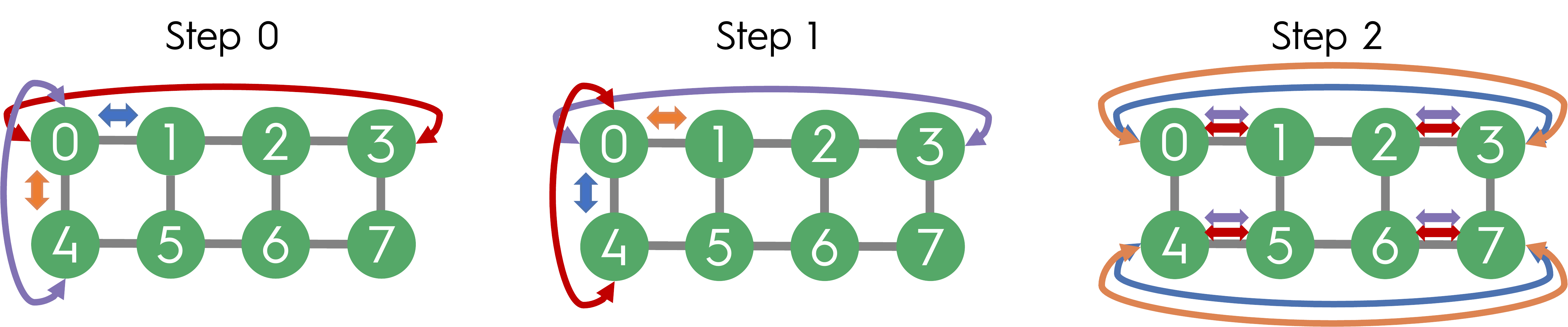}
    \caption{Multiport Swing collective on a 2x4 torus.}
    \label{img:rectangular}
\end{figure}

While there is no difference in latency and bandwidth deficiency, the congestion deficiency increases compared to a square torus. In practice, it will be somewhere between that of an equivalent 1D torus (in the worst case) and of a perfectly square $D$-dimensional torus with the same number of nodes. The actual congestion deficiency depends on the relative differences between the dimensions. We denote with $d_{min}$ the smallest dimension and with $d_{max}$ the largest one and we consider, as a worst case, a $d_{min} \times d_{min} \times \ldots \times d_{max}$ torus.

The algorithm uses all the ports for the first $D \cdot \log_2(d_{min})$ steps and behaves as it runs on a $D$-dimensional torus with $d_{min}^D$ nodes. After that, it used only two ports for the remaining steps, behaving as it runs on a 1D torus, but starting from step $s=\log_2{d_{min}}$ and on data of size $\frac{n}{2 \cdot 2^{D\log_2{d_{min}}}} = \frac{n}{2 d_{min}^D}$. We denote the congestion deficiency of this second phase as:
\begin{align}\label{eq:congdef:rect}
    \begin{split}
        \Xi_{Q} &= \frac{1}{2 \cdot d_{min}^D} \sum_{s=\log_2{d_{min}}}^{\log_2(d_{max}) - 1} \frac{2^{s+ 1} - (-1)^{s +  1}}{3 \cdot 2^{s + 1 - \log_2{d_{min}}}} \approx \\
        &\approx \frac{1}{6 \cdot d_{min}^D} \sum_{s=\log_2{d_{min}}}^{\log_2(d_{max}) - 1} 2^{\log_2{d_{min}}} = \\
         &= \frac{(\log_2{(d_{max}/d_{min})})}{6 \cdot d_{min}^{D-1}}
    \end{split}
\end{align}

We can then approximate the congestion deficiency for rectangular tori by summing the one for square tori to the one in Eq.\ref{eq:congdef:rect} (that is $0$ for square tori). Generally, the higher the ratio between the largest and the smallest dimension, the more steps the algorithm executes not using all the available $2D$ ports. Nevertheless, we show in Sec.~\ref{sec:evaluation:rectangular} that Swing still outperforms state-of-the-art algorithms by up to 3x, except on very large allreduce ($\geq$ 128MiB).

\section{Experimental Evaluation}\label{sec:evaluation}
We evaluate the performance of the Swing algorithm on several torus and torus-like networks by comparing it with the best state-of-the-art algorithms described in Sec.~\ref{sec:multidbackground}. We implemented all these algorithms in the \textit{Structural Simulation Toolkit} (SST~\cite{SST}), a packet-level network simulator. We simulate networks with 400Gb/s links with 100ns latency and 300ns of per-hop packet processing latency~\cite{sensi-slingshot,hoefler2023datacenter}. %

Since each node has $2D$ ports, the maximum injection bandwidth of a node is $2 \cdot D \cdot 400 Gb/s$. Also, in all the plots, we show the \textit{goodput}, i.e., how many bytes are reduced per time unit. Because the allreduce needs to send at least twice the number of bytes in the vector~\cite{BARNETT1995191}, the peak goodput is half the injection bandwidth (i.e., $D \cdot 400 Gb/s$). 

In the following, we analyze the performance on 2D square (Sec.~\ref{sec:evaluation:2d}) and rectangular tori (Sec.~\ref{sec:evaluation:rectangular}), higher-dimensional tori (Sec.~\ref{sec:evaluation:34d}), and other torus-like topologies (Sec.~\ref{sec:evaluation:hx}). Eventually, we summarize the results (Sec.~\ref{sec:evaluation:summary}).

\subsection{Performance on 2D Square Torus}\label{sec:evaluation:2d}
In Fig.~\ref{img:torus_4096}, we show the performance evaluation on a 64x64 2D torus with \num{4096} nodes. In the main plot, we show the goodput of the allreduce for different vector sizes, with each line representing a different algorithm. For the Swing algorithm, for each size we only report the best between the latency- and bandwidth-optimal versions, and we annotate the point where we switch from the latency-optimal to the bandwidth-optimal algorithm with a large dot. We do something similar for recursive doubling as well. 

\begin{figure}[h]
        \begin{center}
            \includegraphics[width=\columnwidth]{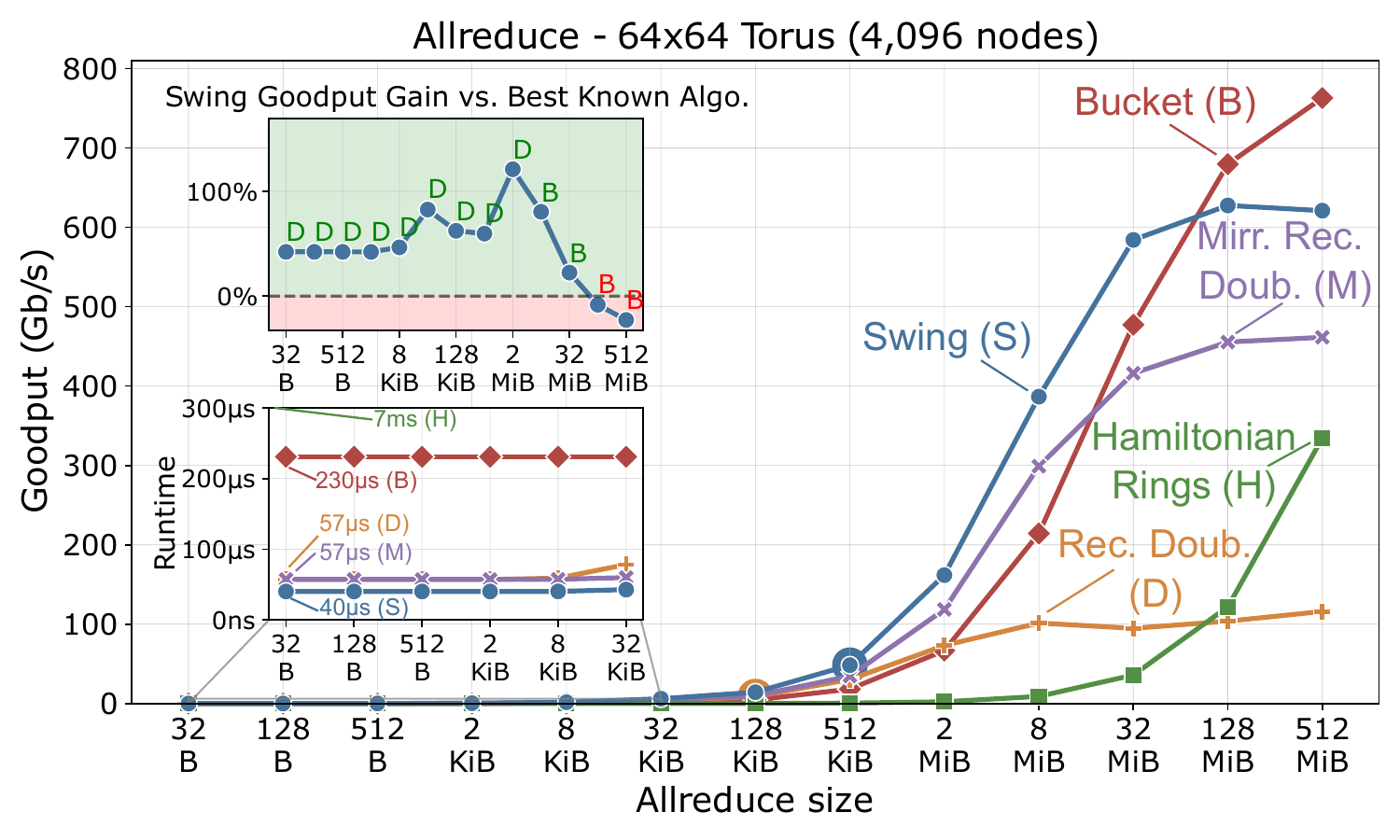}
        \end{center}
        \caption{Goodput of allreduce algorithms on a 64x64 2D Torus topology with \num{4096} nodes. The small plot in the bottom left reports runtime for small allreduce (ranging from 32B to 32KiB). The runtime for 32B allreduce is annotated using the same one-letter labels as in the main plot. The top left inner plot shows the goodput gain of Swing compared to the best-known state-of-the-art algorithm, and the letters on top of each datapoint denote the name of the best-known algorithm.}
        \label{img:torus_4096}
\end{figure}

The zoomed-in plot in the bottom left shows the runtime of each algorithm for allreduce on small vectors (from 32B to 32KiB). We also denote the runtime for 32B allreduce for each algorithm, using the same one-letter labels used in the main plot. Because the runtime of the ring algorithm for small vectors is orders of magnitude larger than the other algorithms, it is always out of scale (the time is nevertheless annotated on the top of the small plot). 

Last, the top left plot shows, for each allreduce size, the goodput gain of Swing compared to the best-known algorithm as a function of the allreduce size. For example, a $100\%$ gain denotes Swing is 2x faster than the best state-of-the-art algorithm. The letter at each data point represents the name of the best-known algorithm. 

For completeness, in this plot, we also show our improved version of recursive doubling, which uses all the ports (denoted as \textit{Mirrored Recursive Doubling}) using the same \textit{plain} and \textit{mirrored} allreduce technique used by Swing (and described in Sec.~\ref{sec:swing:multid}). Swing consistently outperforms our mirrored recursive doubling at any size due to the lower congestion deficiency, and we thus exclude it from the comparison and from the subsequent results.

By analyzing the results, we observe that Swing outperforms all the other allreduce algorithms for vectors ranging from 32B to 32MiB due to the lower latency deficiency compared to the ring and bucket algorithms and the lower bandwidth deficiency compared to the latency-optimal recursive doubling algorithm. We observe more than 2x improvement over the recursive doubling algorithm for 2MiB allreduce. 

The bucket algorithm performs better than Swing starting from 128MiB due to its lower congestion deficiency, which compensates for the higher latency deficiency on large vectors. This is instead not the case for the ring algorithm, characterized by a higher latency deficiency than the bucket algorithm. Moreover, we observe that on 512MiB, Swing achieves around $77\%$ of the peak goodput, which is what we would expect from our modeling. Indeed, a congestion deficiency of $1.19$ on a 2D torus (see Table~\ref{tab:deficiency:torus}) means Swing can reach at most $81\%$ of peak goodput.

By analyzing the two small inner plots, we observe up to $50\%$ improvement for small vectors ($\leq$ 32KiB) compared to the latency-optimal recursive doubling algorithm. This is partly due to the lower congestion deficiency, but mostly to the shorter distance between communicating nodes, which reduces the latency $\alpha$ (although we did not explicitly model it so as not to burden the notation). We observe the highest goodput gains (around $120\%$) for 2MiB vectors. This is indeed the sweet spot where the recursive doubling algorithm performs poorly, and the performance of the bucket and ring algorithms is still severely affected by their higher latency deficiency.

\subsubsection{Scaling}
We then analyze the performance of Swing on 2D torus of different sizes. We show in Fig.~\ref{img:scaling} the goodput gain of Swing over the best-known algorithm at each allreduce size and for networks ranging from \num{64} to \num{16384} nodes. We observe that Swing outperforms state-of-the-art algorithms regardless of the network size, for up to 32MiB allreduce. 

\begin{figure}[hbtp]
        \begin{center}
            \includegraphics[width=\columnwidth]{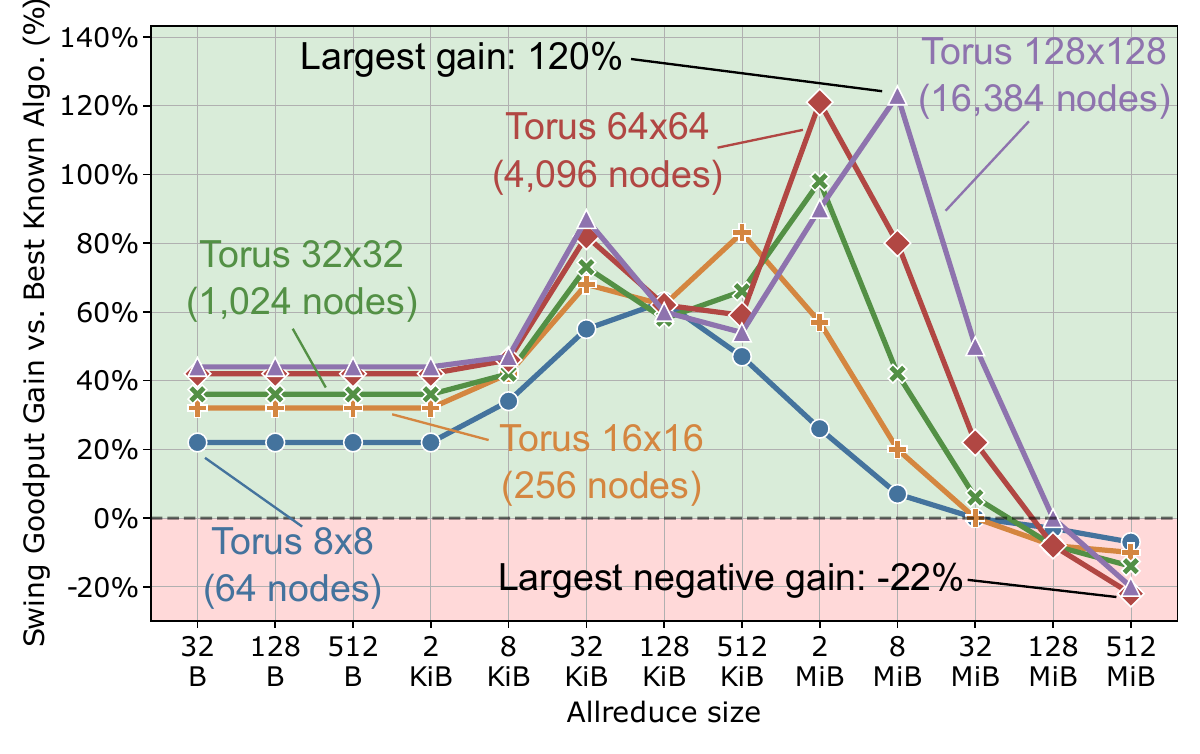}
        \end{center}
        \caption{Swing goodput gain on square torus networks ranging from \num{64} to \num{16384} nodes.}
        \label{img:scaling}
\end{figure}

Moreover, we can see that the maximum Swing gain increases when increasing the network size. Indeed, the larger the network, the larger the impact of latency deficiency on performance. As a consequence, both the bucket and ring algorithm performance decreases when increasing the number of nodes (and the Swing gain increases). 

Nevertheless, the bucket algorithm outperforms Swing for large vectors ($\geq$ 128MiB). Indeed, when increasing the network size, Swing congestion deficiency increases. However, as estimated in Table~\ref{tab:deficiency:torus} and as shown in the figure, on 2D torus, we expect at most a negative gain of around $20\%$ (i.e., a peak bandwidth of around $80\%$) regardless of network size.

\subsubsection{Bandwidth Impact}
To analyze Swing performance for different network bandwidths, we show in Fig.~\ref{img:scaling_bw} the goodput gain for 8x8 torus networks with bandwidth ranging from $100$ Gb/s to $3.2$ Tb/s. We observe consistent gains over the best-known state-of-the-art algorithm regardless of the network bandwidth. 

For low bandwidths, the relative impact of bandwidth and congestion deficiencies on performance is higher, and the gain of Swing over recursive doubling for small messages increases. At higher bandwidth, the relative impact of congestion deficiency is lower, and the maximum gain of Swing for small allreduce decreases. At the same time, however, Swing is not outperformed anymore by the bucket algorithm for large allreduce. For example, on $3.2$ Tb/s networks, Swing outperforms all the other algorithms even for 512MiB allreduce. %

\begin{figure}[hbtp]
        \begin{center}
            \includegraphics[width=\columnwidth]{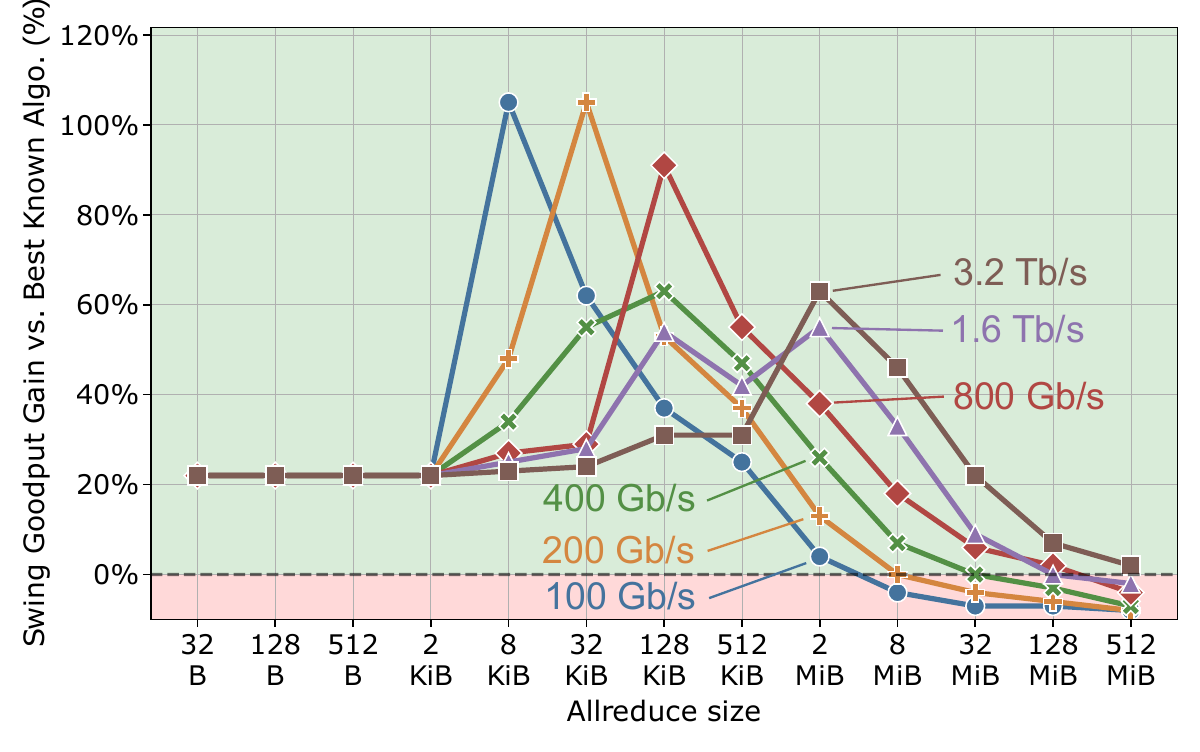}
        \end{center}
        \caption{Swing goodput gain on 8x8 torus networks with network bandwidth ranging from $100$ Gb/s to $3.2$ Tb/s.}
        \label{img:scaling_bw}
\end{figure}

\subsection{Performance on Rectangular Tori}\label{sec:evaluation:rectangular}
As discussed in Sec.~\ref{sec:swing:rectangular}, the congestion deficiency of the Swing algorithm increases if the torus dimensions are not all equal, proportionally to the ratio between the sizes of the largest and smallest dimension. The ring algorithm is instead unaffected by the shape of the torus, because the Hamiltonian rings span over all the nodes (as long as the conditions discussed in Sec.~\ref{sec:background:ring} are satisfied). 

On the other hand, the shape of the torus negatively impacts the bucket algorithm's latency deficiency.
Indeed, if some torus dimensions are larger than the others, some of the $2D$ concurrent collectives might move from dimension $i$ to dimension $i+1$, whereas there are still collectives running on dimension $i+1$. We show this through an example in Fig.~\ref{img:bucket_rect}, where at step $1$, one bidirectional ring moves from the vertical to the horizontal dimension while the other bidirectional ring is still running.

\begin{figure}[htpb]
        \begin{center}
            \includegraphics[width=.7\columnwidth]{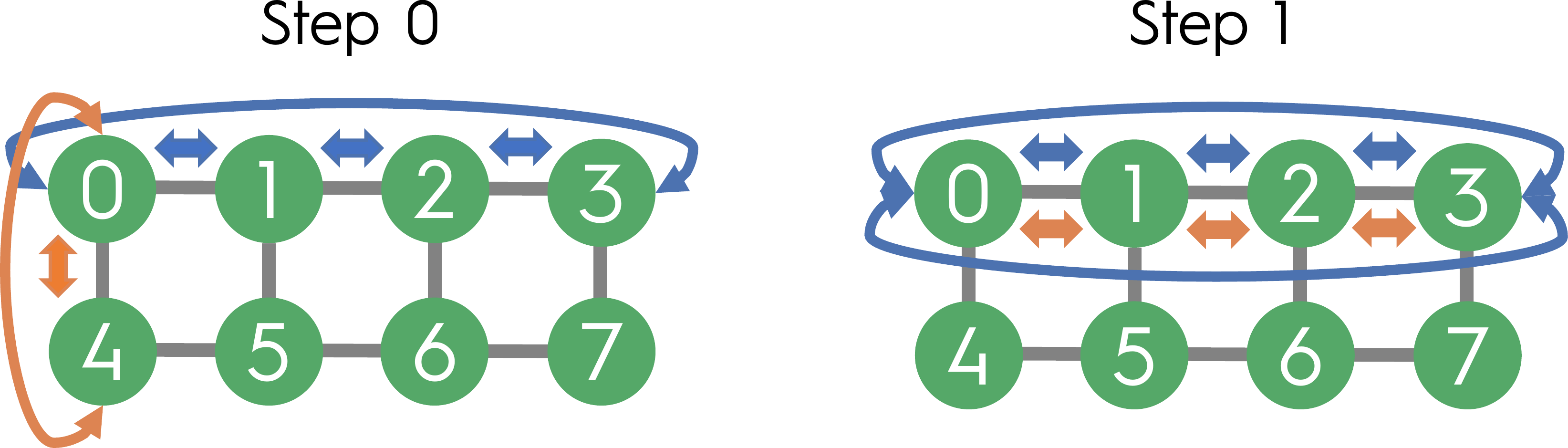}
        \end{center}
        \caption{First steps of the Bucket algorithm on a 2x4 torus.}
        \label{img:bucket_rect}
\end{figure}

This is usually detrimental to performance, and it is better if all the collectives synchronously move from one dimension to the next one~\cite{10.1145/2686882}. Thus, each step is completed only after all the collectives running on the largest dimension are completed. Hence, if $d_{max}$ is the size of the largest dimension, the latency deficiency becomes $\defl = \frac{2D \cdot d_{max}}{\log_2{p}}$. Bandwidth and congestion deficiencies are instead unaffected. In a nutshell, this means that the latency deficiency of the bucket algorithm is the same as that for a $d_{max} \times d_{max} \cdots \times d_{max}$ torus.

For these reasons, we show in Fig.~\ref{fig:diff_rect} the goodput of the different algorithms for different torus networks, all with \num{1024} nodes but with different rectangular shapes. First, we observe that, as expected, the ring algorithm is unaffected by the shape of the torus, and outperforms both the bucket and Swing algorithms for allreduce larger than 512MiB. 
On the other hand, the latency deficiency of the bucket algorithm increases proportionally to the ratio between the largest and the smallest dimensions, reducing its performance for small and medium vectors. 
This is visible, for example, by analyzing how the goodput for large allreduce decreases when moving from a 64x16 to a 256x4 torus. 

Last, Swing performance decreases compared to a square torus due to its higher congestion deficiency for rectangular torus networks. Nevertheless, we observe that Swing still outperforms the other algorithms up to 32MiB regardless of the network shape (up to 3x on the 128x8 and 256x4 torus).

\begin{figure}[hbtp]        
    \centering\subfigure{\includegraphics[width=\columnwidth]{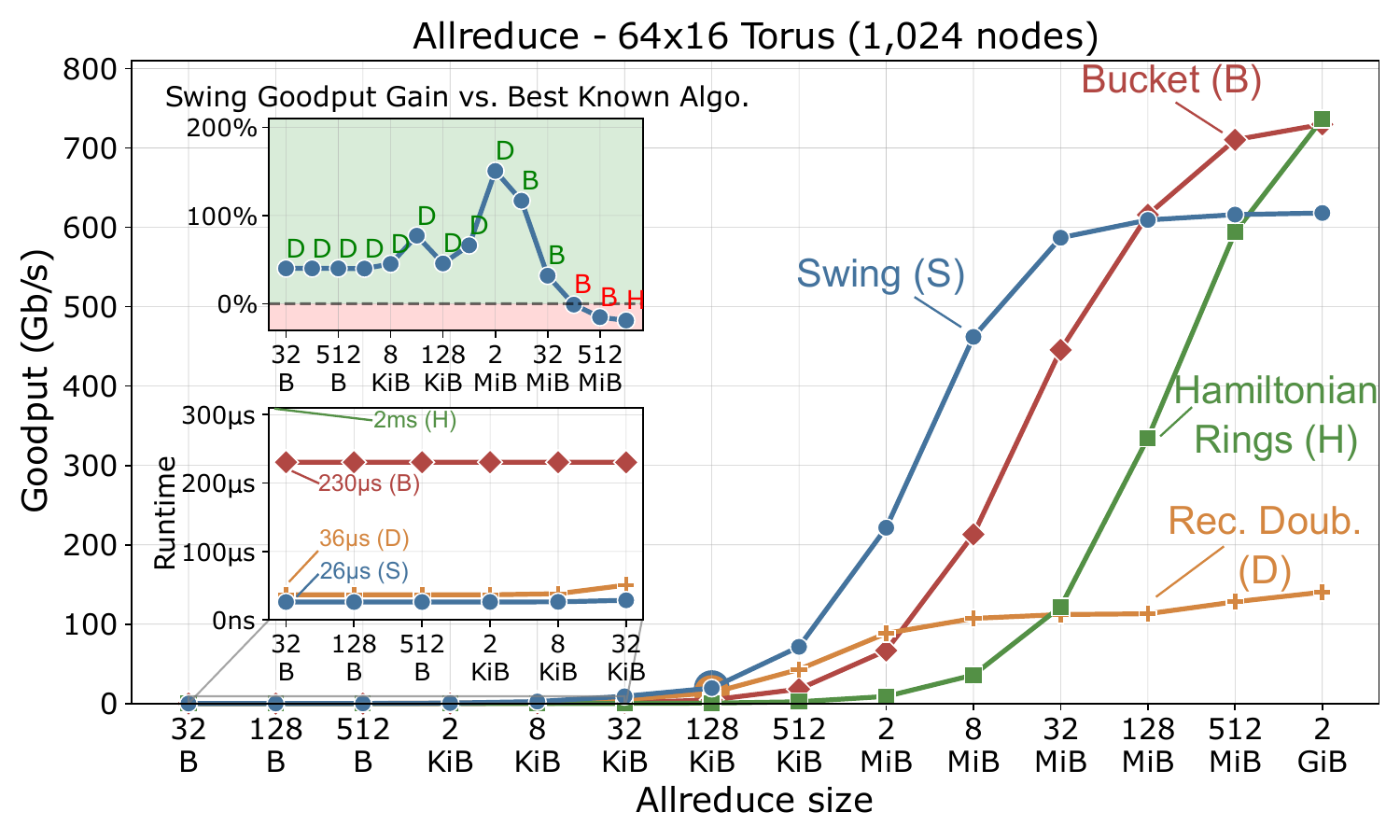}} 
    \centering\subfigure{\includegraphics[width=\columnwidth]{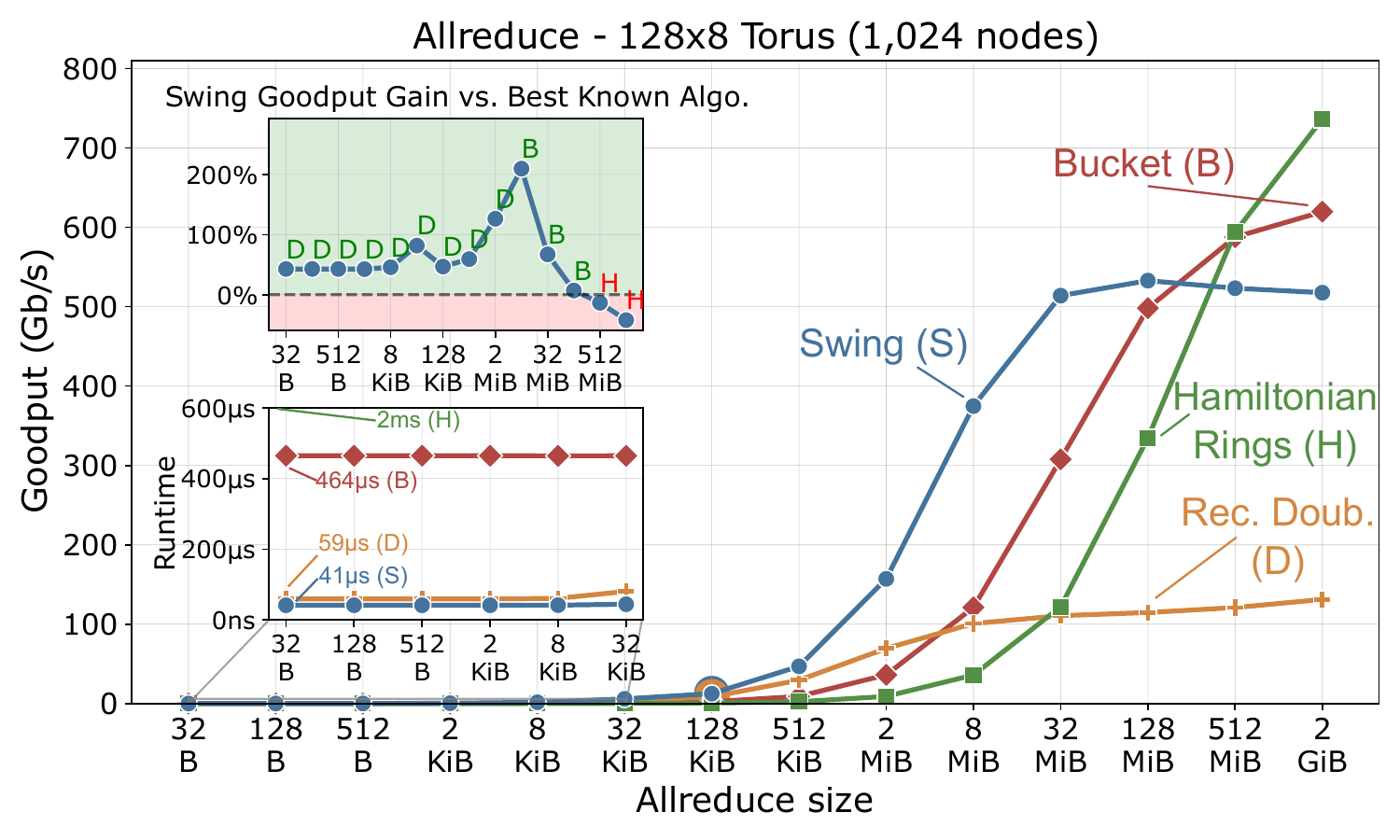}} 
    \centering\subfigure{\includegraphics[width=\columnwidth]{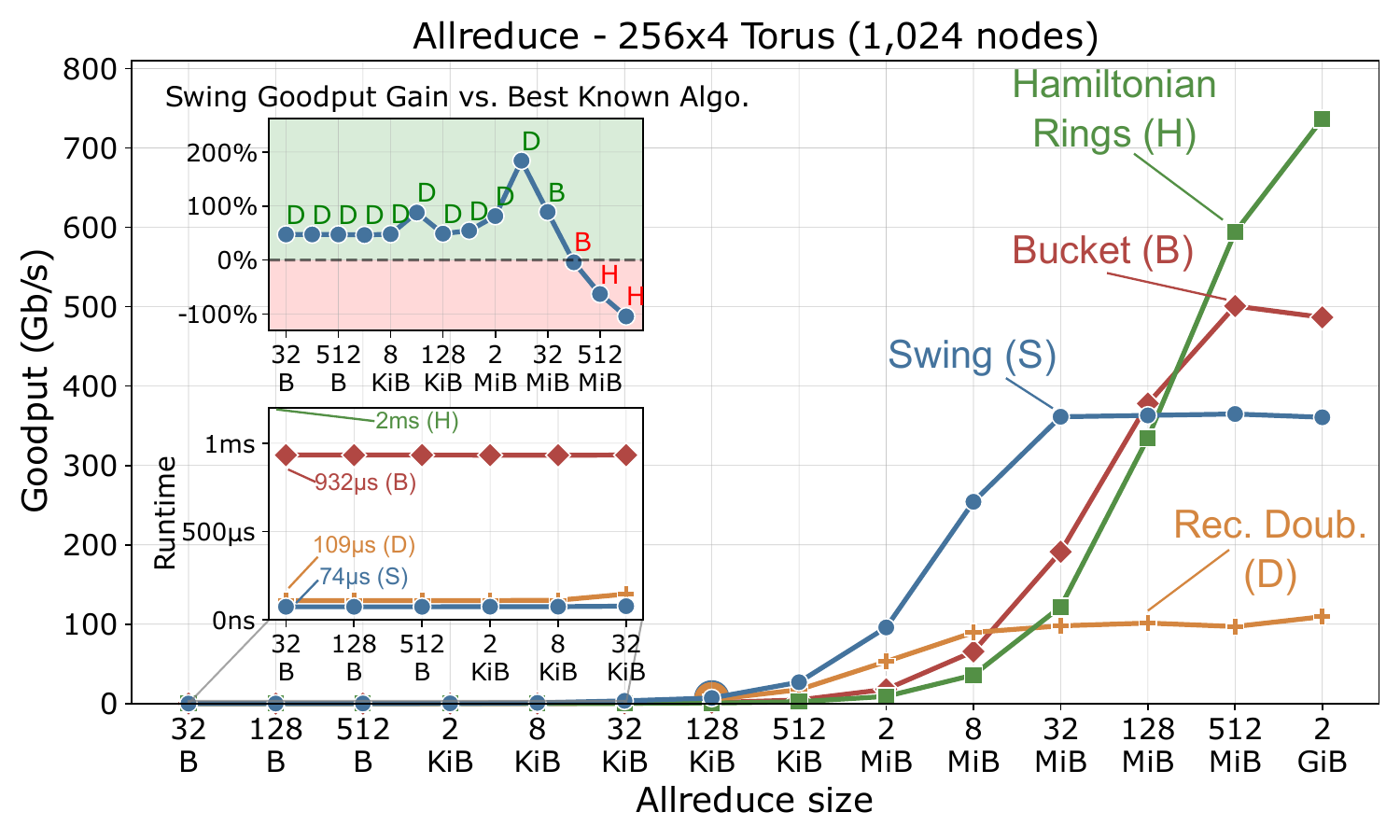}} 
    \caption{Goodput on 2D torus with \num{1024} nodes and different rectangular shapes: 64x16, 32x8, and 256x4.}
    \label{fig:diff_rect}
\end{figure}

\subsection{Performance for 3D and 4D Torus}\label{sec:evaluation:34d}
As discussed in Sec.~\ref{sec:swing:multid} and summarized in Table~\ref{tab:deficiency:torus}, the performance of the allreduce algorithm for multidimensional torus also depends on the number of dimensions. Thus, we evaluate the performance of the different allreduce algorithms on $8^2$, $8^3$, and $8^4$ torus networks. 

\begin{figure}[htpb]
    \centering\subfigure{\includegraphics[width=\columnwidth]{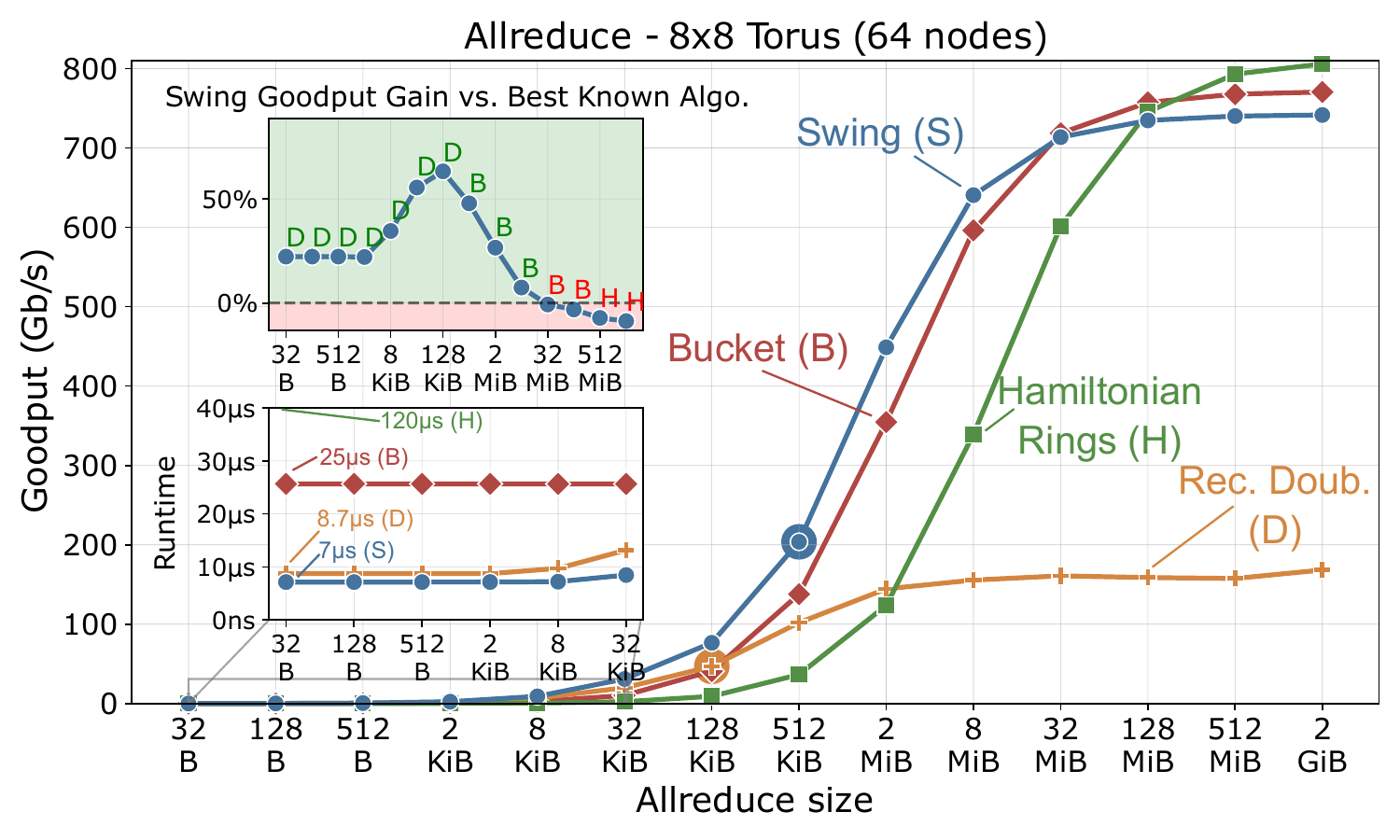}} 
    \centering\subfigure{\includegraphics[width=\columnwidth]{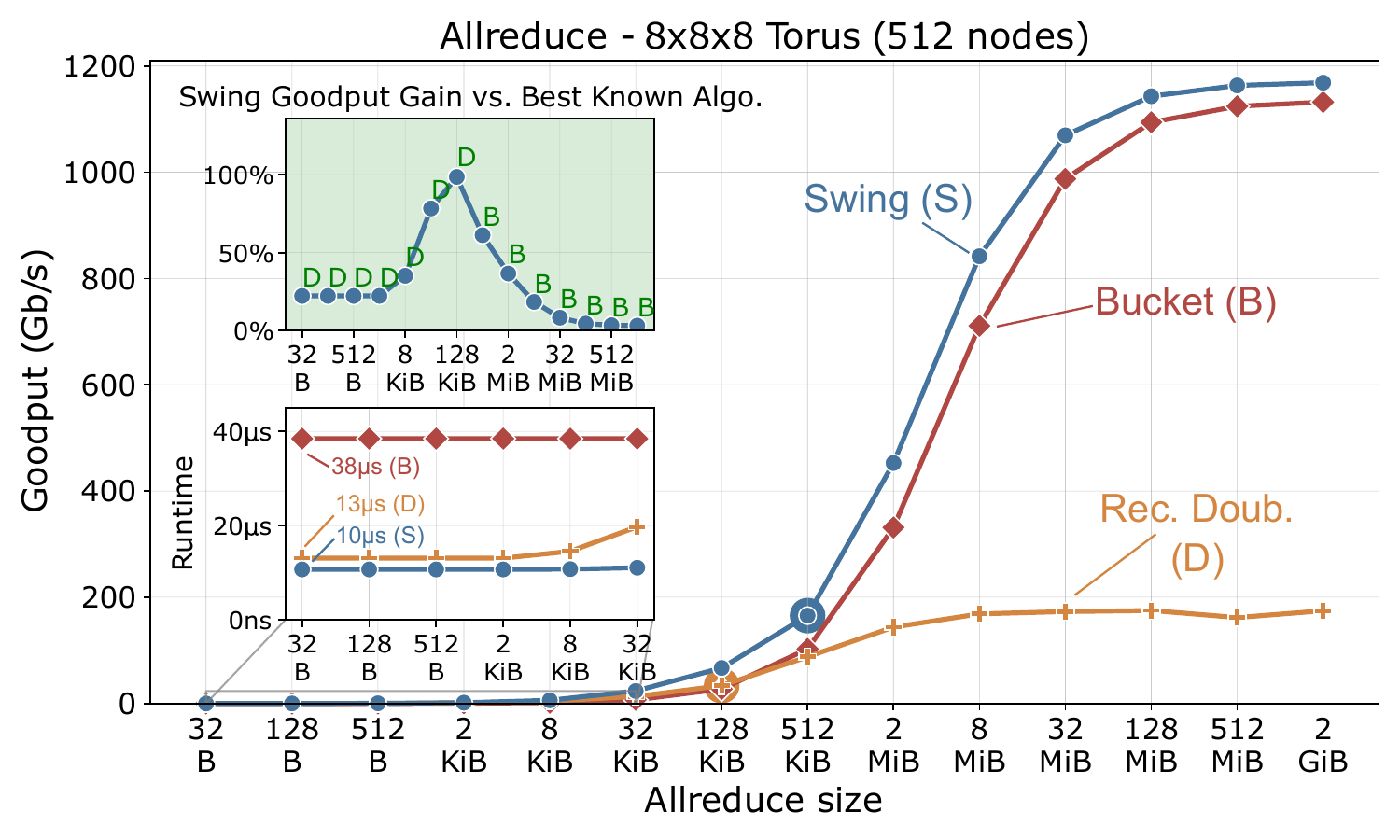}} 
    \centering\subfigure{\includegraphics[width=\columnwidth]{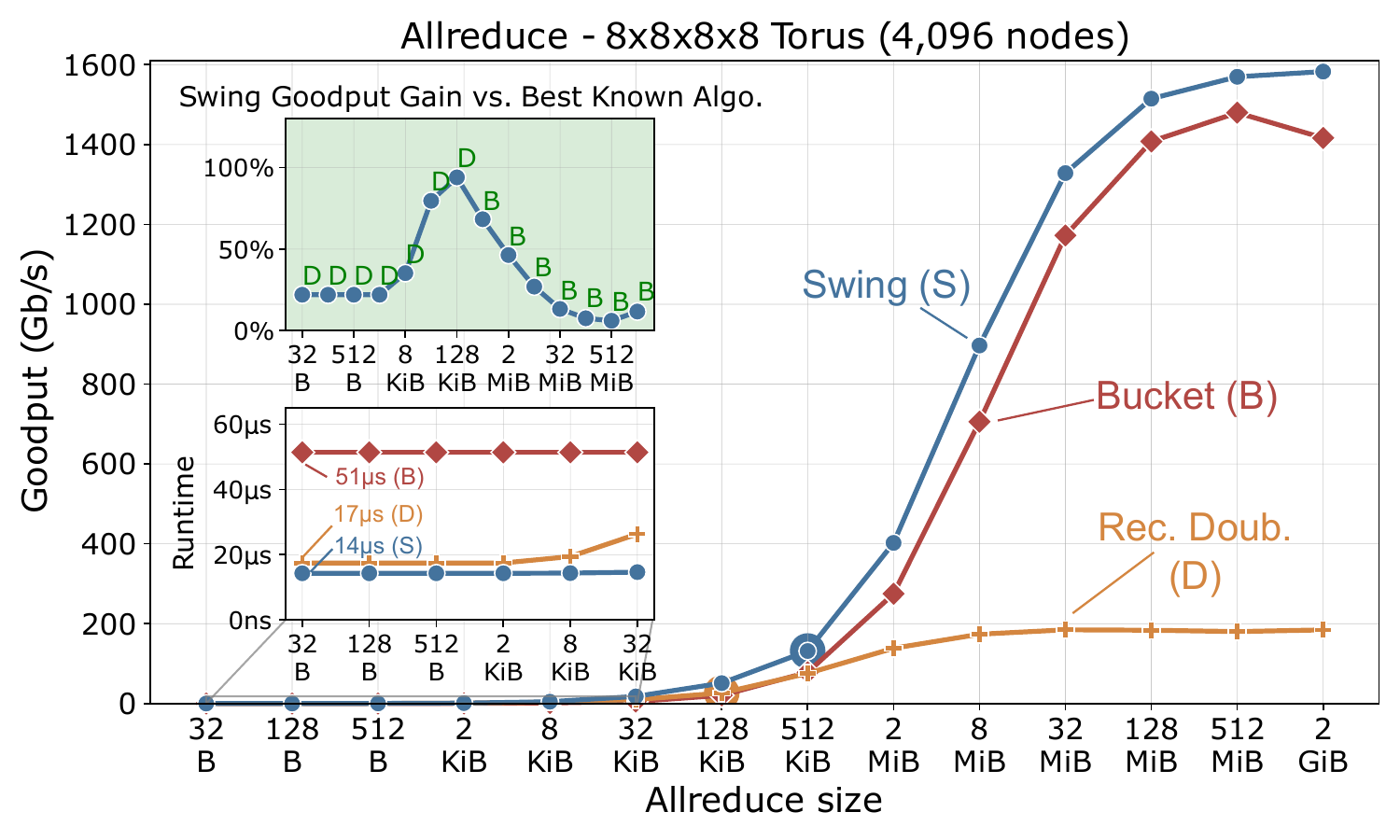}} 
    \caption{Goodput on higher-dimensional torus networks: 2D 8x8, 3D (8x8x8), and 4D (8x8x8x8).}
    \label{fig:diff_sizes}
\end{figure}

We report the evaluation result in Fig.~\ref{fig:diff_sizes}. We do not include the Hamiltonian ring algorithm in the 3D and 4D torus results since it only works for 2D torus networks. 
When increasing the number of dimensions, the goodput gain of Swing increases because, as shown in Table~\ref{tab:deficiency:torus} and discussed in Sec.~\ref{sec:swing:multid}, the congestion deficiency drops to $3\%$ on 3D torus and to $0.8\%$ on 4D torus. Consequently, \textit{{for 3D and 4D torus networks, Swing outperforms by up to 2x all existing algorithms on allreduce ranging from 32B to 2GB}}.

\subsection{Performance on Torus-Like Topologies}\label{sec:evaluation:hx}
Some topologies like HammingMesh~\cite{10.5555/3571885.3571899} and HyperX~\cite{10.1145/1654059.1654101,10.1145/3295500.3356140} extend torus by adding additional links, thus increasing the network bisection bandwidth. Seen from a different perspective, those extra links allow distant nodes to communicate crossing fewer hops, decreasing Swing congestion deficiency.

\subsubsection{Performance on HammingMesh}
HammingMesh~\cite{10.5555/3571885.3571899} groups nodes into square \textit{boards}. Each board is a 2D mesh, and nodes on the same column (or row) located at the edge of the boards are connected together using fat trees. Due to its higher performance and flexibility compared to a torus a similar topology is used, for example, to interconnect TPUv4 devices~\cite{jouppi2023tpu}. Because of the extra links, the congestion deficiency of Swing on a HammingMesh is lower than that on a 2D torus. Moreover, for a fixed number of nodes, having smaller boards increases the number of extra (fat tree) links and, thus, decreases the congestion deficiency.

\begin{figure}[h]
        \begin{center}
            \includegraphics[width=\columnwidth]{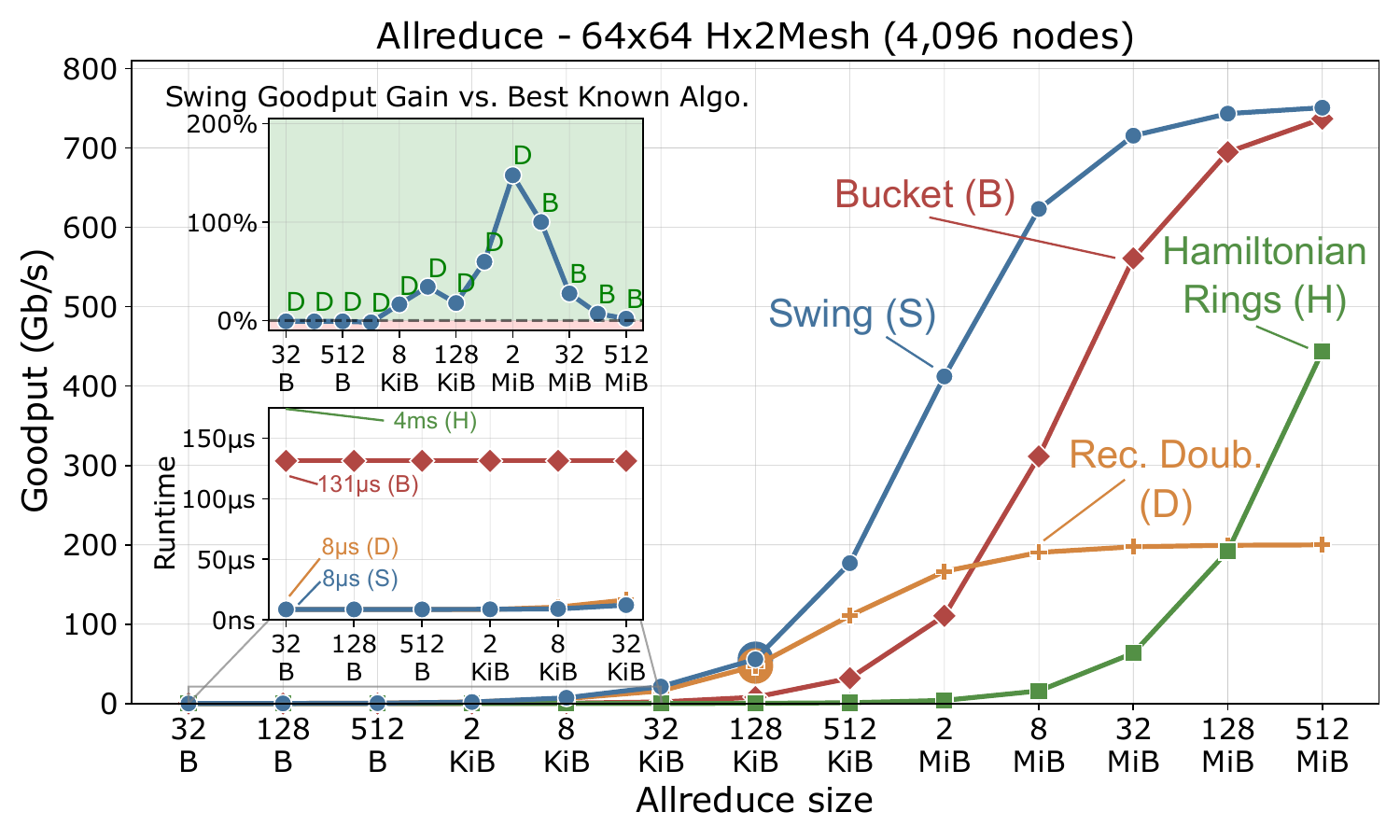}
        \end{center}
        \caption{Goodput on a \num{4096} nodes Hx2Mesh.}
        \label{img:hx2_4096}
\end{figure}

We show in Fig.~\ref{img:hx2_4096} the performance of the different algorithms for a Hx2Mesh network with \num{4096} nodes (2x2 boards arranged in a 32x32 configuration). For such configuration, Swing outperforms the state-of-the-art algorithms at any size, up to 2.5x for 2MiB allreduce. Moreover, because of the lower congestion deficiency, we observe how the peak Swing performance is higher compared to a 2D torus with the same number of nodes (Fig.~\ref{img:torus_4096}). Last, we also observe a runtime reduction for all the algorithms for small vectors, since nodes on the same board on HammingMesh are connected through PCB traces, with lower latency than optical network cables.

In Fig.~\ref{img:hx4_4096} we instead report the results for a HammingMesh with the same number of nodes (\num{4096}), but using 4x4 boards (4x4 boards arranged in a 16x16 configuration). 
This configuration is a middle point between the torus topology and the Hx2Mesh since it has more extra links than a torus, but fewer than a Hx2Mesh. Due to fewer extra links compared to Hx2Mesh, on Hx4Mesh Swing has a higher congestion deficiency, as we can see starting from 128MiB allreduce. 

\begin{figure}[h]
        \begin{center}
            \includegraphics[width=\columnwidth]{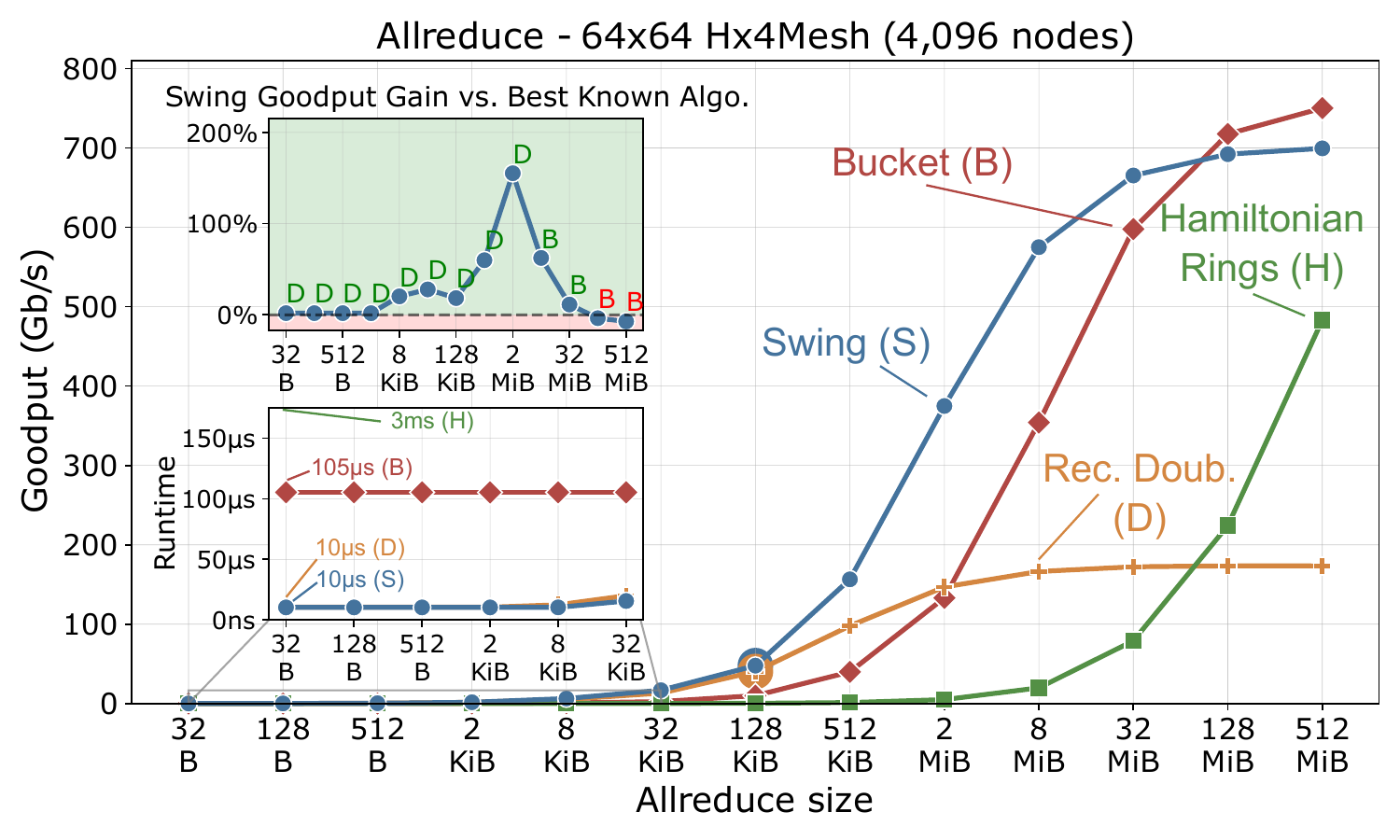}
        \end{center}
        \caption{Goodput on a \num{4096} nodes Hx4Mesh.}
        \label{img:hx4_4096}
\end{figure}

\subsubsection{Performance on 2D HyperX}
Last, we report in Fig.~\ref{img:hyperx} the performance on a \num{4096} 2D HyperX topology~\cite{10.1145/1654059.1654101} (which can be seen as a HammingMesh with 1x1 boards~\cite{10.5555/3571885.3571899}). HyperX connects each node to every node in the same row and column. Because in Swing each node communicates only with nodes on the same row (or the same column), on HyperX Swing does not experience any congestion deficiency. We indeed observe from the plot that Swing outperforms all the other algorithms at any allreduce size. Moreover, the maximum goodput gain increases from 2.5x in Hx2Mesh and Hx4Mesh, to 3x in HyperX.

\begin{figure}[h]
        \begin{center}
            \includegraphics[width=\columnwidth]{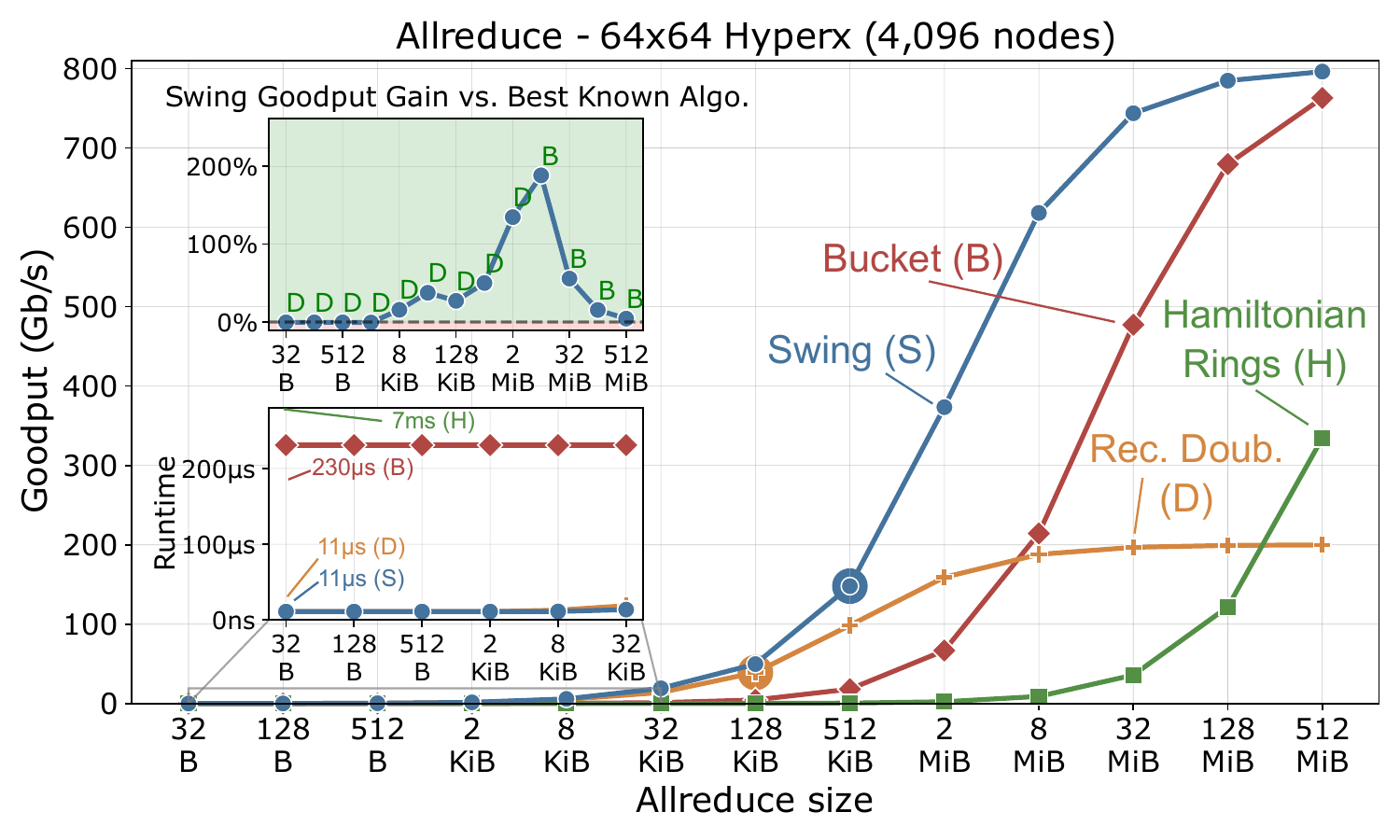}
        \end{center}
        \caption{Goodput on a \num{4096} nodes HyperX.}
        \label{img:hyperx}
\end{figure}

\subsection{Summary}\label{sec:evaluation:summary}
Last, we summarize all the presented results in Fig.~\ref{img:summary_gain}, where we report the distribution of Swing goodput gain over the best algorithm at each message size for the different scenarios we analyzed. We show data for allreduce sizes up to 512MiB since these are the sizes practically used in existing HPC~\cite{8665758} and machine learning workloads~\cite{10.14778/3415478.3415530} and larger allreduce are usually split into smaller ones to overlap communication and computation better. 

For each box, the triangle shows the median across the allreduce sizes. The left and right sides of the box represent the first (\textit{Q1}) and third (\textit{Q3}) quartile, respectively. The left whisker denotes the smallest point larger than $Q1 - 1.5 \cdot (Q3-Q1)$, whereas the right whisker the biggest point smaller than $Q3 + 1.5 \cdot (Q3-Q1)$. Empty dots outside the whiskers are considered outliers.

\begin{figure}[h]
        \begin{center}
            \includegraphics[width=\columnwidth]{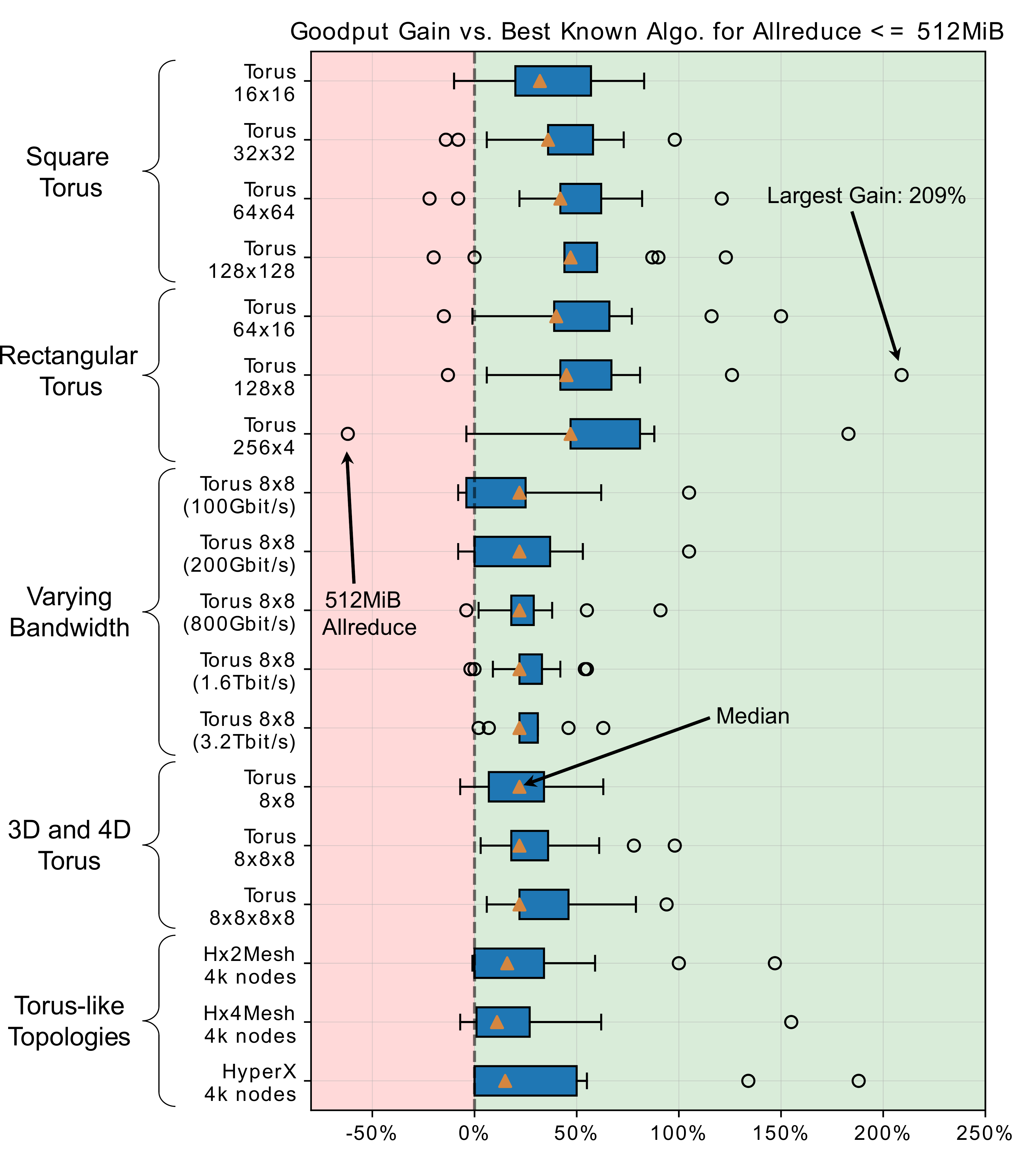}
        \end{center}
        \caption{Summary of Swing goodput gain on different topologies.}
        \label{img:summary_gain}
\end{figure}

First, by observing the performance on square torus networks, we see that Swing median and maximum goodput gain increases with the network size (going from 8x8 to 128x128 torus). As discussed in Sec.~\ref{sec:evaluation:2d}, the maximum negative gain also increases with network size, but is at most around $20\%$ and only occurs for large allreduce ($\geq$ 128MiB).

For rectangular torus networks, Swing median and maximum goodput gain increase proportionally to the ratio between the maximum and minimum-sized dimension  (i.e., going from 64x16 to 256x4 torus), up to 3x for 128x8 torus. For 256x4 torus, Swing performs around $60\%$ worse than the ring algorithm on 512MiB allreduce.

Increasing the network bandwidth has a positive effect on Swing goodput gain. Although the maximum goodput gain decreases, for higher bandwidths, Swing outperforms all the state-of-the-art algorithms at any allreduce size. Regardless of the bandwidth, the median Swing goodput gain across different allreduce sizes stands around $25\%$.

Swing goodput gain also increases with the number of dimensions (moving from 8x8 to 8x8x8x8 torus) since each node has additional one-hop distant nodes. Because in the reduce-scatter nodes communicate first with their neighbors, and halve the data size at each step, Swing can further reduce the transmitted data size before nodes need to communicate with more distant ones. On one side, this increases the maximum goodput gain up to 2x. On the other side, it outperforms all the other algorithms at any size. 

Last, Swing provides consistent performance gains even on other torus-like topologies such as HammingMesh and HyperX. When comparing HammingMesh with a torus with the same number of nodes (i.e., 64x64), we observe higher goodput gains on HammingMesh, due to the extra links compared to the torus, which helps in reducing Swing congestion deficiency. Also, on both Hx2Mesh and HyperX, Swing outperforms the other algorithms regardless of the allreduce size.

Overall we observe a median goodput ranging between $20\%$ and $50\%$, and a maximum goodput gain of 3x. This underlines the advantages of using the Swing allreduce algorithm on torus and torus-like networks of any shape and number of dimensions and at any network bandwidth.

\section{Discussion}
\paragraph{Extension to Other Collectives}
Although in this work we focus on allreduce (and, indirectly, on reduce-scatter and allgather), a similar approach can be adopted for other collective operations. Namely, Swing can replace the recursive doubling algorithm for all those collectives where it is used (e.g., broadcast and reduce).

\paragraph{Routing Impact}
In all our experiments, we used minimal adaptive routing. I.e., packets are forwarded on the least congested shortest path. Because in all the analyzed algorithms, each node only communicates with others on the same row/column, traffic is evenly distributed across links. Sending packets on non-shortest paths would unnecessarily increase network traffic and decrease performance.

\paragraph{Swing Performance on Full-Bandwidth Topology}
On full-bandwidth topologies (e.g., non-blocking fat trees), both Swing and recursive doubling will not have any congestion deficiency, and we expect them to have the same performance.

\section{Conclusions}
Due to the relevance of torus networks in high-performance machine learning systems, 
in this work, we presented \textit{Swing}, a new allreduce algorithm for torus networks. To motivate Swing design, we modeled the latency, bandwidth, and congestion deficiencies, of Swing and of the best state-of-the-art algorithms. Our modeling highlights the shortcomings of existing allreduce algorithms for torus networks, especially for small- and medium-sized allreduce.

We then presented the Swing design. Swing performs a logarithmic number of steps and transmits a minimal number of bytes. To reduce the impact of the torus low bisection bandwidth, Swing shortcuts the torus, reducing the distance between communicating nodes and, thus, the congestion deficiency.
Last, we extensively evaluated Swing performance and compared it against the best state-of-the-art algorithms, for different node counts, network bandwidths, shapes, number of dimensions, and topologies. Our evaluation shows improvements up to 3x on all practically used allreduce sizes.%

\section*{Acknowledgments}
We thank the anonymous NSDI reviewers and
our shepherd, Philip Levis, for their constructive feedback. This work was supported by Sapienza University under the SEED-2022 and "Progetti Grandi 2023" funding schemes; by UrbanTwin, an ETH Board Joint Initiatives project; and by EuroHPC-JU under grant agreement RED-SEA, No 055776. Daniele De Sensi is a member of the \textit{Gruppo Nazionale Calcolo Scientifico - Istituto Nazionale di Alta Matematica} (GNCS-INdAM).

\bibliographystyle{plain}
\bibliography{main}

\appendix
\section{Correctness Proof}\label{sec:proof}
In this section we first prove the correctness of the Swing algorithm when the number of nodes is a power of two (Sec.~\ref{sec:proof:power}) and then extend it for an arbitrary number of nodes (Sec.~\ref{sec:proof:nonpower}).

\subsection{Power of Two Number of Nodes}\label{sec:proof:power}
For clarity reasons, we prove the correctness of the latency-optimal algorithm (both the latency- and bandwidth-optimal algorithms perform the same communication pattern, although nodes transmit different data). To prove the algorithm's correctness, we need to prove that the data transmitted by each node eventually reaches all the other nodes. For example, on a 1D torus with $p$ nodes, at step $0$, node $0$ communicates with node $1$ (which aggregates the received data with its own). At step $1$, node $0$ communicates with node $-1 \bmod p = p-1$, whereas node $1$ communicates with node $2$. Thus, we can say that at step $1$ the data sent from $0$ reached nodes  $\{1, 2, p-1\}$ ($2$ has been reached \textit{indirectly} through node $1$). 

Because the number of reached nodes doubles at each step, and because we perform $\log_2{p}$ steps, the data sent from any given node would eventually reach $p-1$ nodes. We need, however, to prove that those $p-1$ nodes are distinct (i.e., that the data sent by each node reach every other node exactly once and is thus never aggregated twice). To do so, we need first to prove a few lemmas.

\begin{lemma}\label{rhoodd}
$\rho(s)$ and $\delta(s)$ are odd $\forall s \in \mathbb{N}$.
\end{lemma}

\begin{proof}
$(-2)^i$ is odd for $i=0$, and even for $i>0$. The sum of even numbers with an odd number is odd.
\end{proof}

\begin{lemma}\label{evenodd}
If $r$ is even, $\pi(r, s)$ is odd, and vice versa.
\end{lemma}

\begin{proof}
An even node $r$ communicates at step $s$ with node $\pi(r, s) = r + \rho(s) \bmod{p}$. 
Because $p$ is a power of two (thus even), and $\rho(s)$ is odd (Lemma~\ref{rhoodd}), 
$\pi(r, s)$ is odd. Vice versa, odd nodes communicate with even nodes.
\end{proof}

If a node $r$ communicates at step $s$ with a node $q = \pi(r, s)$, and $q$ communicates with a node $z = \pi(q, h)$ at step $h > s$, we say that $s$ indirectly reached node $z$. Because even nodes only communicate with odd nodes (and vice versa), if $r$ is even, we can rewrite:
\begin{equation}
z = \underbrace{\overbrace{(r + \rho(s) \bmod{p})}^{q = \pi(r, s)} - \rho(h) \bmod{p}}_{\pi(q, h)} = r + \rho(s) - \rho(h) \bmod{p} \nonumber
\end{equation}
I.e., the sign behind $\rho(s)$ alternates between positive and negative, starting from positive. In general, an even node $r$ can reach through a sequence of $k$ steps $\{s_0 < s_1 < s_2 < \ldots < s_{k-1}\}$ a node $q$, with:
\begin{align}
\begin{split}
q &= r + \rho(s_0) - \rho(s_1) + \rho(s_2) - \ldots \bmod{p} = \\
  &= (r + \sum_{i=0}^{k-1} -1^{i}\rho(s_i)) \bmod{p} \nonumber
\end{split}
\end{align}
The same applies if $r$ is odd, by replacing $-1^i$ with $-1^{i+1}$.

\begin{lemma}\label{oddsteps}
Even nodes reach (directly or indirectly) odd nodes through an odd number of steps $k$. Odd nodes reach (directly or indirectly) even nodes through an odd number of steps $k$.
\end{lemma}

\begin{proof}
This stems from Lemma~\ref{evenodd}. If $r$ is even and $k$ is odd, then $q = (r + \sum_{i=0}^{k-1} -1^i \rho(s_i)) \bmod{p}$ is odd because $\rho(s)$ is always odd. Similarly, if $r$ is odd and $k$ is odd, $q$ is even.
\end{proof}

\begin{lemma}\label{lemma:modzero}
Given $k$ integers $\{e_0 < e_1 < \ldots < e_{k-1}\}$, with $e_{k-1} \leq \log_2(p) - 1$, then $-p < \sum_{i=0}^{k-1} (-2)^{e_i} < p$.
\end{lemma}
\begin{proof}
    We have $\sum_{i=0}^{k-1} (-2)^{e_i} \leq \sum_{i=0}^{k-1} 2^{e_i} < 2^{e_{k-1} + 1} \leq p$. Similarly, $\sum_{i=0}^{k-1} (-2)^{e_i} \geq -\sum_{i=0}^{k-1} 2^{e_i} > -(2^{e_{k-1} + 1}) \geq -p$.
\end{proof}

\begin{theorem}\label{swingcorrectness}
On a 1D torus, if a node $r$ at step $s$ communicates with node $\pi(r, s)$ (defined in Eq.~\ref{eq:pi}), it will reach (directly or indirectly) all the other $p-1$ nodes in $log_2(p)$ steps (with $p$ power of two).
\end{theorem}

\begin{proof}
We need to prove that, a unique sequence of $k$ steps $\{s_0 < s_1 < \ldots s_{k-1}\}$ exists by which a given node $r$ reaches a node $q$. We prove this by contradiction, and we will prove it by assuming that $r$ is even and $q$ is odd (the proof for the other cases is analogous and only requires changing the signs before the $\rho$ terms).  Assume that there are two different sequences of steps $\{s_0 < s_1 < \ldots s_{k-1} \leq \log_{2}(p) - 1\}$ and $\{t_0 < t_1 < \ldots t_{h-1} \leq \log_{2}(p) - 1\}$ of $k$ and $h$ steps respectively (both $k$ and $h$ are odd from Lemma~\ref{oddsteps}), so that:
\begin{align}\label{eq:newproof:first}
\begin{split}
q &= r + \rho(s_0) - \rho(s_1) + \rho(s_2) - \ldots + \rho(s_{k-1}) \bmod{p} = \\
  &= r + \rho(t_0) - \rho(t_1) + \rho(t_2) - \ldots + \rho(t_{h-1}) \bmod{p} \\
\end{split}
\end{align}

By expanding the first of the two sequences we have:
\begin{align}
\begin{split}
q 
  &= r + \sum_{i=0}^{s_0} (-2)^i - \sum_{i=0}^{s_1} (-2)^i + \ldots + \sum_{i=0}^{s_{k-1}} (-2)^i \bmod{p} \\
  &= r + \sum_{i=0}^{s_0} (-2)^i + \sum_{i=s_1 + 1}^{s_2} (-2)^i + \ldots + \sum_{i=s_{k-2} + 1}^{s_{k-1}} (-2)^i  \bmod{p} \nonumber
\end{split}
\end{align}

By expanding similarly the second assignment in Eq.~\ref{eq:newproof:first}, we have that
the two sequences exist if:
\begin{align}\label{eq:newproof:exp}
    \begin{split}
        & \sum_{i=0}^{s_0} (-2)^i + \ldots + \sum_{\mathclap{i=s_{k-2} + 1}}^{s_{k-1}} (-2)^i \equiv \\ & \equiv 
          \sum_{i=0}^{t_0} (-2)^i + \ldots + \sum_{\mathclap{i=t_{k-2} + 1}}^{t_{h-1}} (-2)^i \Mod{p}
    \end{split}
\end{align}

From Lemma~\ref{lemma:modzero}, we know that both sides are in the range $(-p, p)$. Thus, the two sides are congruent only if: i) they have the same sign and are equal, or; ii) they have different signs, and by summing $p$ on the negative side, we get the positive side. Since each side is the sum of distinct powers of $-2$, case i) is only possible if the two sequences of steps are equal. To prove that case ii) is impossible, let us consider the case where the left side is negative (the other case is analogous). Because $p = 2^a$ for some $a \in \mathbb{N}$, and because $2^a = (-2)^{a}$ (if $a$ is even\footnote{If $a$ is odd, $p = 2^a = (-2)^{a+1} + (-2)^{a}$, and the same considerations still hold.}), Eq.~\ref{eq:newproof:exp} becomes:

\begin{align}
    \begin{split}
    & \sum_{i=0}^{s_0} (-2)^i + \ldots + \sum_{\mathclap{i=s_{k-2} + 1}}^{s_{k-1}} (-2)^i + (-2)^{a} = \\ &=
    \sum_{i=0}^{t_0} (-2)^i + \ldots + \sum_{\mathclap{i=t_{k-2} + 1}}^{t_{h-1}} (-2)^i\nonumber
    \end{split}
\end{align}
However, because both sides are sums of distinct powers of $-2$, they can be equal only if the two sequences of steps are equal, which implies that there must be a $t_i$ such that $t_i = a$. This is impossible because the number of steps can be at most equal to $\log_{2}(p) - 1 = a - 1$. 

We thus proved by contradiction that there are no two sequences of steps leading to the same node and that, at each step, each node reaches only nodes that it did not already reach.
\end{proof}

\subsection{Non-Power of Two Number of Nodes}\label{sec:proof:nonpower}
The correctness proof in Sec.~\ref{sec:proof:power} assumes $p$ is a power of $2$ (needed by the last part of Theorem~\ref{swingcorrectness}). If $p$ is not a power of $2$, the theorem only holds until the second-last step. If $p'$ is the largest power of $2$ smaller than $p$, in the second last step, the data sent by each node reached $p' - 1$ nodes. Thus, we need to guarantee that in the last step: i) no nodes receive data it already received; ii) each node reaches the remaining $p - p' - 1$ nodes. To guarantee property i), it is enough for each node to pre-compute the blocks $b_i$ it will send at each step and if it would send a block twice, send that only in the last step.

To guarantee property ii), it is enough to prove that $\pi(r, s) = \pi(g, s) \Leftrightarrow g = s$. Indeed, if no two nodes reach the same node in the last step, then each node has reached each other node once. First, if $r$ and $g$ are both even, we need to prove that $r + \rho(s) \equiv g + \rho(s) \pmod{p}$. This implies $r - g \equiv 0 \pmod{p}$. However, because we have $r < p$ and $g < p$, this is only possible if $r = g$. The proof for odd $r$ and $g$ is analogous. Then, if $r$ is even and $g$ is odd, we have $r + \rho(s) \equiv g - \rho(s) \pmod{p}$, which implies $2 \rho(s) \equiv g - r \pmod{p}$. However, we know from Lemma~\ref{rhoodd} that $\rho(s)$ is always odd, and thus $2\rho(s)$ is even. Because $r$ is even and $g$ is odd, $g-r$ is odd. Thus, if $p$ is even, this can only hold for $r=g$.

When $p$ is even but not a power of $2$, it is enough for each node not to send the same data block twice, thus not increasing the deficiency compared to the power of two case.

Last, if $p$ is odd, we might have $\pi(r, s) = \pi(g, s)$ even if $r \neq g$, and also Lemma~\ref{evenodd} does not hold anymore. Consequently, some nodes might not reach all the nodes and, at a given step, might receive data from more nodes simultaneously, thus decreasing the performance. For this reason, we run the algorithm on $p-1$ nodes, with the \textit{"odd"} node sending data directly to each of the other nodes.

\end{document}